\documentclass[12pt, draftclsnofoot,onecolumn]{IEEEtran}
\usepackage[utf8]{inputenc} 
\usepackage[T1]{fontenc}    
\usepackage{hyperref}       
\usepackage{url}            
\usepackage{booktabs}       
\usepackage{amsfonts}       
\usepackage{nicefrac}       
\usepackage{microtype}      
\usepackage{xcolor}         
\usepackage{tikz,pgfplots,graphicx}
\usepackage{amstext,cite}
\usetikzlibrary{plotmarks,spy,backgrounds}
\pgfplotsset{compat=newest}
\usepackage{caption}
\usepackage{subcaption}
\usepackage{pstricks}
\usepackage{bbm}
\usepackage{amsmath,amssymb,amsthm}
\usepackage{enumitem}
\usepackage{algorithm}
\usepackage{algorithmic}

\usepackage{cleveref}

\definecolor{mycolor1}{rgb}{0.85098,0.32549,0.09804}%
\definecolor{mycolor2}{rgb}{0.92941,0.69412,0.12549}%
\definecolor{mycolor3}{rgb}{0.00000,0.44706,0.74118}%
\definecolor{mycolor4}{rgb}{0.30196,0.74510,0.93333}%
\definecolor{mycolor5}{rgb}{0.46667,0.67451,0.18824}%
\definecolor{mycolor6}{rgb}{0.4940,0.1840,0.5560} %

\theoremstyle{plain}

\theoremstyle{definition}

\theoremstyle{remark}
\newtheorem{remark}{Remark}

\newcommand{\T}{{\sf T}}

\newcommand{\RR}{{\mathbb{R}}}

\newcommand{\EE}{{\mathbb{E}}}

\newcommand{\C}{\mathbf{C}}

\newcommand{\X}{\mathbf{X}}

\newcommand{\x}{\mathbf{x}}

\newcommand{\bt}{\mathbf{t}}

\newcommand{\one}{\mathbf{1}}
\newcommand{\I}{\mathbf{I}}

\newcommand{\bmu}{\boldsymbol{\mu}}

\newcommand{\bbeta}{\boldsymbol{\beta}}

\definecolor{RED}{rgb}{0.7,0,0}
\definecolor{BLUE}{rgb}{0,0,0.69}
\definecolor{GREEN}{rgb}{0,0.6,0}
\definecolor{PURPLE}{rgb}{0.69,0,0.8}

\newtheorem{Theorem}{Theorem}

\newtheorem{Proposition}{Proposition}

\begin{document}

\title{{An Achievable and Analytic Solution to Information Bottleneck for Gaussian Mixtures}}

\author{Yi~Song,~\IEEEmembership{Student~Member~IEEE,}
Kai~Wan,~\IEEEmembership{Member~IEEE,} 
Zhenyu~Liao,~\IEEEmembership{Member~IEEE,} 
Hao~Xu,~\IEEEmembership{Member~IEEE,} 
Giuseppe~Caire,~\IEEEmembership{Fellow~IEEE, }
Shlomo~Shamai~(Shitz),~\IEEEmembership{Fellow~IEEE, }
\thanks{
A short version of this paper   was accepted  by   the 2024 IEEE International Symposium on Information Theory. 
}
\thanks{The work of Y. Song and G. Caire have been supported by DFG Gottfried Wilhelm Leibniz-Preis. 
The work of K.~Wan is partially supported by the National Natural Science Foundation of China (via fund NSFC-12141107) and the Interdisciplinary Research Program of HUST (2023JCYJ012).
The work of Z.~Liao is partially supported by the National Natural Science Foundation of China (via fund NSFC-62206101), the Guangdong Key Lab of Mathematical Foundations for Artificial Intelligence Open Fund (OFA00003), the Fundamental Research Funds for the Central Universities of China (2021XXJS110), and Key Research and Development Program of Guangxi (GuiKe-AB21196034).
The work of H. Xu has been  supported in part by the European Union's Horizon 2020 Research and Innovation Programme under Marie Skłodowska-Curie Grant No. 101024636 and the Alexander von Humboldt Foundation. The work of S. Shamai has been supported by the
European Union's Horizon 2020 Research and Innovation Programme with grant agreement No. 694630.}
\thanks{Y. Song and G. Caire are with Faculty of Electrical Engineering and Computer Science, Technical University of Berlin, Berlin, Germany (email: yi.song@tuberlin.de, caire@tu-berlin.de). }
\thanks{K. Wan and Z. Liao are with School of Electronic Information and Communications, Huazhong University of Science and Technology, Wuhan, China (email: kai\_wan@hust.edu.cn, zhenyu\_liao@hust.edu.cn). }
\thanks{H. Xu is with the Department of Electronic and Electrical Engineering, University College London, London WC1E7JE, UK (e-mail: hao.xu@ucl.ac.uk).}
\thanks{S. Shamai (Shitz) is with the Viterbi Electrical Engineering Department, Technion–Israel Institute of Technology, Haifa 32000, Israel (e-mail: sshlomo@ee.technion.ac.il)}
}


\maketitle

\begin{abstract}
In this paper, we study a remote source coding scenario in which binary phase shift keying (BPSK) modulation sources are corrupted by additive white Gaussian noise (AWGN). 
An intermediate node, such as a relay, receives these observations and performs additional compression to balance complexity and relevance. 
This problem can be further formulated  as an information bottleneck (IB) problem with Bernoulli sources and Gaussian mixture observations. However,  no  closed-form solution exists for this IB problem. 
To address this challenge, we propose a unified achievable scheme that employs three different compression/quantization strategies for intermediate node processing by using two-level quantization, multi-level deterministic quantization, and soft quantization with the hyperbolic tangent ($\tanh$) function, respectively. 
In addition, we extend our analysis to the vector mixture Gaussian observation problem and explore its application in machine learning for binary classification with information leakage. 
Numerical evaluations show that the proposed scheme has a near-optimal performance over various signal-to-noise ratios (SNRs), compared to  the Blahut-Arimoto (BA) algorithm, and has better performance than some existing numerical methods such as the information dropout approach.
Furthermore, experiments conducted on the realistic MNIST dataset also validate the superior classification accuracy of our method compared to the information dropout approach.
\end{abstract}


\IEEEpeerreviewmaketitle
\section{Introduction}
\label{sec:intro}
\subsection{Introduction of IB and its applications in communications}
The information bottleneck (IB) serves as a fundamental framework widely used in both machine learning and information theory to understand and regulate the flow of information within a data processing system. Introduced by Tishby et {\it al.}~\cite{tishby2000information}, the IB problem can be formulated as extracting information from a target random variable $Y$ through an observation $X$ that is correlated with $Y$. This is achieved by establishing the Markov chain $Y \longrightarrow X \longrightarrow T$, where $T$ extracts the information from the observation $X$. The core idea of the IB is to wisely balance the tradeoff between two competing objectives in constructing $T$:
\begin{itemize}
    \item \textbf{Complexity} (or compression) that measures the information required to represent the observation $X$, so that $T$ is a \emph{compact} representation of the observation.
    \item \textbf{Relevance} (or prediction) that measures the information retained in the compressed representation to make accurate predictions about the target variable $Y$, so that $T$ is an informative representation of $Y$.
\end{itemize}
These objectives are typically evaluated by the mutual information between the observation and the compressed representation $I(X; T)$, as well as between the compressed representation and the target variable $I(Y; T)$. The IB problem seeks the optimal conditional probability $P_{T|X}$ by maximizing the relevance $I(Y; T)$ with constrained complexity $I(X; T)$.

Due to its mathematical complexity, the optimal solution for the IB problem was only derived in closed-form for binary symmetric or Gaussian  sources~\cite{zaidientropy2020}, i.e., $X$ and $Y$ are both binary or both Gaussian. 
In the general case, however, the solution of the IB problem relies exclusively on numerical algorithms. For example, a numerically optimal solution can be achieved using the Blahut-Arimoto (BA) algorithm for the IB problem~\cite{tishby2000information}. Extending the BA algorithm, \cite{Hassanpour2017overview} presents several alternative iterative algorithms based on clustering techniques or deterministic quantization methods.
Furthermore, an alternative approach proposed in \cite{alemi2017deep} involves the use of neural networks to establish a lower bound for the Lagrangian IB problem based on samples of $(X, Y)$ pairs.

The IB problem has also found widespread applications in various fields such as communications and machine learning (refer to~\cite{zaidientropy2020, goldfeld2020information} for more details on the application of IB).  
It has been proven in~\cite{dobrushin1962information, witsenhausen1975conditional, courtade2013multiterminal} that 
the IB problem is essentially equivalent to the remote source coding problem with logarithm loss distortion measure~\cite{courtade2013multiterminal}. The authors in~\cite{sanderovich2008communication} have established the connection between operational meaning of the IB problem and relay networks, where the relay with oblivious processing could not directly decode messages from the received signals. This work was then extended to 
 scenarios with multiple sources and relays for cloud radio access networks (C-RANs)~\cite{aguerri2019capacity}. Other studies~\cite{xu2021information, xuhao2021information, xu2022distributed, song2023distributed} have explored similar relay-based setups, specifically under Rayleigh fading channels. These scenarios require relays to consider channel state information when forwarding signals due to the coupling between received signals and channels. The IB problem provides crucial insights and techniques for optimizing data compression in such distributed communication environments.

\subsection{Applications of IB in machine learning}
The IB approach has been widely used in supervised, unsupervised, as well as representation machine learning (ML) tasks (such as inference, prediction, classification, and clustering)~\cite{bengio2013representation, achille2018information} to characterize or explain how relevant information/representations $T$ can be extracted from observations $X$ about a target $Y$, where the two mutual information $I(Y;T)$ and $I(X;T)$ in the IB approach represent the empirical relevance and complexity, respectively. 
Thus, solving the IB problem in a ML context naturally leads to a good tradeoff between fitting the training data and generalizing to unseen test data, which is the ultimate goal of ML~\cite{bishop2006pattern}.
It has been believed, for example, that IB is an efficient way to control generalization error in deep neural networks (DNNs), and that IB provides insights in understanding how neural networks learn to extract relevant features from data and to regularize models for better generalization~\cite{shwartz2017opening, kim2021distilling, rishby2015deep, saxe2019information}. 
In addition, the IB framework can be directly used a metric for constructing more efficient DNN models, by minimizing redundancy between adjacent layers, measured by mutual information, rather than through traditional strategies such as pruning, quantization, and knowledge distillation~\cite{dai2018compressing}.
Nonetheless, from a ML theoretical perspective, much less is known about the optimal IB solution, nor its impact on the generalization performance of the ML model, even for the most fundamental Gaussian mixture model (GMM).
In this paper, we reveal an interesting connection between the IB approach and the binary GMM classification problem with information leakage, in which case IB aims to discover a compressed yet informative representation of the GMM input, so as to achieve the minimal misclassification rate  under limited privacy leakage.

\subsection{Main contributions}
\begin{figure}[ht!]
\centering
\begin{tikzpicture}[node distance = 0.03\textwidth]
\tikzstyle{neuron} = [circle, draw=black, fill=white, minimum height=0.05\textwidth, inner sep=0pt]
\tikzstyle{rect} = [rectangle, rounded corners, minimum width=0.1\textwidth, minimum height=0.05\textwidth,text centered, draw=black, fill=white]
    \node [neuron] (neuron1) {{$Y^n$}};
     \node [right of=neuron1, xshift=0.15\textwidth, neuron] (Y) {$X^n$};
    \node [right of=Y, xshift=0.15\textwidth, rect] (Enc) {Encoder};
    \node [right of=Enc, xshift=0.3\textwidth, rect] (Dec) {Decoder};
     \node [right of=Dec, xshift=0.15\textwidth, neuron] (recon) {$\widehat{Y}^n$};
     \draw [->,line width=1pt] (neuron1) -- (Y);
      \draw [->,line width=1pt] (Y) -- (Enc);
     \draw [->,line width=1pt] (Enc) -- node[above] {\small{$f_n(X^n) \in \{1, 2, ..., 2^{nR}\}$}}(Dec);
     \draw [->,line width=1pt] (Dec) -- (recon);
\end{tikzpicture}
\caption{The system diagram of the remote source coding theory.}  
\label{RSC}
\end{figure}
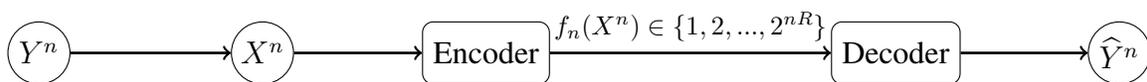
In this paper, we first consider a remote source coding problem with i.i.d.\@ Binary Phase Shift Keying (BPSK) modulation inputs, as illustrated in Fig.~\ref{RSC}. The modulated signal is sent through a  Gaussian additive noise (AWGN) channel. An intermediate node, such as a relay, receives the observation and performs further compression to achieve the optimal tradeoff between complexity and relevance. When the  distortion measure is log-loss, 
to characterize the rate-distortion region for this remote source coding problem is equivalent to solve the IB problem with a Bernoulli source and a Gaussian mixture observation. 
The main contribution of this paper is to 
provide  achievable and   analytic solutions  for this IB problem.  More precisely, 
\begin{itemize}
    \item
  To address the challenge of finding a closed-form solution to the mixture Gaussian IB problem, we propose three analytically achievable schemes that employ different compression/quantization  strategies: two-level quantization, multi-level deterministic quantization, and soft quantization with the $\tanh$ function. Each approach excels in a different region of the tradeoff curve, providing insight into their performance characteristics.
In numerical evaluations,     we compare the proposed  schemes with the numerical solution using the Blahut-Arimoto (BA) algorithm, which can be seen as the approximate optimal solution. Extensive   numerical results under different signal-to-noise ratio (SNR) show that the gap to the BA  algorithm is limited. Furthermore, our proposed  schemes outperform the numerical information dropout approach \cite{achille2018information}.
    \item We extend our proposed achievable schemes to tackle the vector mixture Gaussian observation IB problem, thereby broadening the applicability of our framework to more complex scenarios.
    \item Finally, we investigate the connection between the IB framework and the binary classification problem with information leakage, where the IB serves to extract a maximally compressed yet informative feature for the classification task, under the constraint of limited privacy leakage.  
   We  extend 
   the proposed schemes for the vector mixture Gaussian observation IB problem to this learning application. Experiments on     the MNIST dataset also show the advantage in performance provided by our schemes compared to the   information dropout method. 
\end{itemize}

\subsection{Notations and organization of the paper}
We denote the upper-case letters as random variables, and lower-case letters as their realizations. For a random variable $X$, calligraphic symbol  $\mathcal{X}$ represents the support of $X$;  we denote 
$\EE[X]$, $H(X)$ and $h(X)$ the expectation, the entropy, and  the differential entropy of $X$, respectively.   
For two random variables $X$ and $Y$, we use $I(X; Y)$ to denote their mutual information. We take the base of the logarithm as $e$.
 We also denote $P_X$ as the probability mass function of $X$, while $p_X$ denotes the probability density function of $X$.  Moreover, $\mathbb{P}(X \in \mathcal{A})$ is denoted as the probability of the event $X \in \mathcal{A}$. 
We use $\mathcal{N}(\mu,\sigma^2)$ for Gaussian distribution with mean $\mu$ and variance $\sigma^2$.
The operator $\lceil\cdot \rceil$ denotes the ceiling function, and $\oplus$ denotes the inclusive `or' operation. 
$\mathbbm{1}_{\{\mathcal{A}\}}$ denotes the indicator function of the condition $\mathcal{A}$, i.e., it gives $1$ when $\mathcal{A}$ is satisfied, and $0$ otherwise.

This paper is organized as follows. 
The system model of the considered IB problem and some preliminary results are introduced in Section~\ref{sec:system}. 
Our main technical results on an achievable closed-form solution to the IB problem is given in Section~\ref{sec: main_results}. 
Extension on the vector mixture Gaussian observation is presented in Section \ref{sec: application}.
Section \ref{sec:connection_ML} discusses the   application of the proposed schemes in machine learning. 
Numerical results are provided in Section~\ref{sec: simulation} to validate the proposed IB scheme, on both synthetic and real-world datasets.
Finally, the conclusion is placed in Section~\ref{sec: conclusion}.

\section{System Model and Preliminary Results}
\label{sec:system}
\subsection{Formulation of the IB Problem}
\label{sub:IB formulation}
In this paper, we consider the remote source coding problem, where the sequences of i.i.d.\@ output from the Binary Phase Shift Keying (BPSK) flow through an additive white Gaussian noise (AWGN) channel.
The intermediate node receives the noisy observations, and performs further compression, e.g., by solving an IB problem, to achieve the optimal tradeoff between the complexity and relevance, for the decoder to estimate the source sequences.

Assume the source $Y^n =(Y_1, Y_2,\ldots,Y_n)$ is drawn i.i.d.\@ from  a symmetric Bernoulli distribution (that is, $Y_i = \pm 1$ with $\mathbb{P}(Y_i=-1) = \mathbb{P}(Y_i=1) = 1/2$ for each $i\in \{1,2,\ldots,n\}$), and the observation $X^n=(X_1,X_2,\ldots,X_n) \in\RR^n$ follows a Gaussian mixture where
\begin{equation}\label{eq:def_model_scalar}
  X_i = \beta Y_i + {\epsilon_i}, \ \forall i\in \{1,2,\ldots,n\}, 
\end{equation}
for some \emph{deterministic} scalar $\beta \in \RR^{+}$ (without loss of generality, we assume that $\beta$ is non-negative) and i.i.d.\@ AWGN ${\epsilon_i} $.
The intermediate node applies an encoding function  $f_{\rm enc}^n(\cdot)$:
\begin{align}
    f_{\rm enc}^n \colon \mathcal{X}^n \longrightarrow \{1, 2, \ldots, 2^{nR}\},
\end{align}
where $R$ represents the coding rate. After receiving  $f^n_{\rm enc}(X^n)$
 the decoder reconstructs ${T}^n$ with alphabet ${\mathcal{T}}^n$ through a decoding function
\begin{align}
    f_{\rm dec}^n\colon \{1, 2, \ldots, 2^{nR} \} \longrightarrow {\mathcal{T}}^n.
\end{align}
Given a distortion requirement $D$, the decoder aims to achieve 
\begin{align}\label{eq:constraint_RSC}
    \mathbb{E}[d_n({T}^n, Y^n)] \leq D,
\end{align}
where $d_n({T}^n, Y^n) = \frac{1}{n} \sum_{i=1}^n d(Y_i, {T}_i)$, under some distortion measure 
$d : {\mathcal{T}} \times \mathcal{Y} \rightarrow \mathbb{R}_+$.

With large enough block length, i.e., $n \rightarrow \infty$,
the infimum of the rate to encode the observations given distortion requirement $D$ is given by ~\cite{hanineq}
\begin{align}\label{eq:RD_function}
    R(D) = \min_{P_{T|X}: \mathbb{E}[d(Y, T)] \leq D} &I(X; T),
\end{align}
where $X|Y \sim \mathcal N( \beta Y, 1)$ with $Y = \pm 1$, $\mathbb{P}(Y=-1) = \mathbb{P}(Y=1) = 1/2 $, and $P_{X,Y,T}=P_{X,Y}P_{T|X}$.
\begin{figure}[htb]
\centering
\begin{tikzpicture}[node distance = 0.03\textwidth]
\tikzstyle{neuron} = [circle, draw=black, fill=white, minimum height=0.05\textwidth, inner sep=0pt]
\tikzstyle{rect} = [rectangle, rounded corners, minimum width=0.1\textwidth, minimum height=0.05\textwidth,text centered, draw=black, fill=white]
    \node [neuron] (neuron1) {{$Y$}};
     \node [right of=neuron1, xshift=0.1\textwidth, neuron] (Y) {$X$};
    \node [right of=Y, xshift=0.1\textwidth, rect] (Enc) {Encoder};
    \node [right of=Enc, xshift=0.1\textwidth, neuron] (U) {$T$};
     \draw [->,line width=1pt] (neuron1) -- (Y);
      \draw [->,line width=1pt] (Y) -- (Enc);
         \draw [->,line width=1pt] (Enc) -- (U);
\end{tikzpicture}
\caption{\small{Diagram of the information bottleneck problem.}}  
\label{IB}
\end{figure}
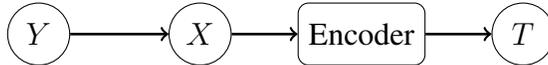

Next, we consider the case where the decoder produces a ``soft'' reconstruction of $Y^n$, i.e., 
the representation variable $T$ is a probability vector over $\mathcal{Y}$. The fidelity of a soft estimate is measured through the log-loss distortion \cite{courtade2013multiterminal}, given as 
\begin{align}
d(t, y) = \log \frac{1}{t(y)},
\end{align}
where $t(y)$ denotes the probability of $T$ evaluated at $T=y$ when given $Y = y$.
In this case, the distortion constraint in \eqref{eq:RD_function} given as $ \mathbb{E}[d(Y, T)] \leq D$ can reduce to $H(Y|T) \leq D$. By noticing that $I(Y;T) = H(Y) - H(Y|T)$, $H(Y)$ is fixed by $P_Y$ (one bit in our case) and therefore minimizing $H(Y|T)$ is equivalent to maximizing $I(Y;T)$ . 
Therefore, the solutions $(R,D)$ of \eqref{eq:RD_function} coincide with that of the IB problem \cite{courtade2013multiterminal} (as illustrated in Fig. \ref{IB}).
 \begin{subequations}
     \begin{align}
    \max_{P_{T|X}} \quad & I(Y; T)  \\
    \textrm{s.t.} \quad &I(X; T) \leq R, \label{eq:obj_constr} 
\end{align}\label{eq:obj_func_1}%
\end{subequations}
where $X|Y \sim \mathcal N( \beta Y, 1)$ with $Y = \pm 1$, $\mathbb{P}(Y=-1) = \mathbb{P}(Y=1) = 1/2 $, and $P_{X,Y,T}=P_{X,Y}P_{T|X}$.
In other words, we are interested in designing the conditional probability $P_{T|X}$ to construct an intermediate representation $T$ of $X$ so that:
\begin{itemize}
  \item[(i)] $T$ contains sufficiently rich information (in the sense that $I(Y; T)$ is large) on the source $Y$, and
  \item[(ii)] the bottleneck constraint is satisfied (with $I(X; T) \leq R$).
\end{itemize}

\subsection{Approximately numerically optimal scheme: Blahut-Arimoto (BA) algorithm}
\label{sec: BA}
A closed-form solution to the IB problem in \eqref{eq:obj_func_1},  beyond the case of jointly Gaussian and symmetric Bernoulli $(X, Y)$, to the best of our knowledge, remains an open problem \cite{goldfeld2020information}.
The Lagrangian form of \eqref{eq:obj_func_1} over the conditional probability $P_{T|X}$, is given by
\begin{align}
    L(\lambda) &= \min_{P_{T|X}} I(X; T) - \lambda 
 I(Y; T), \label{eq:obj_func_2} 
\end{align}
where, according to \cite{blahut1972computation}, ${\lambda}^{-1}$ can be defined as the slope of the curve of $I(Y; T)$ versus $R$, i.e., ${\lambda}^{-1} \overset{\Delta}{=} \frac{\partial I(Y; T)}{\partial R}$.
Thus $L(\lambda)$ can represent   the tradeoff between the mutual information $I(Y; T)$ and $I(X; T)$. 

Following the computation on rate-distortion function by the well-known Blahut-Arimoto (BA) algorithm \cite{blahut1972computation},  Tishby \textit{et. al.\@} in \cite{tishby2000information} proposed to apply an iterative algorithm  to solve the IB problem \eqref{eq:obj_func_2} numerically by initializing $P_{T|X}(t|x)$ with the randomly generated normalized probability $P_{T|X}^{\rm init}(t|x)$ and the algorithm updates three probabilities iteratively:
\begin{subequations}\label{eq:numerical_solution}
\begin{align}
  P_T(t) &= \sum_{x \in \mathcal{X}} P_{T|X}(t|x) P_X(x), \\
    P_{Y|T}(y|t) &= \frac{\sum_{x \in \mathcal{X}} P_{X|Y}(x|y) P_{T|X}(t|x) P_Y(y)}{P_T(t)} , \\
     P_{T|X}(t|x) &= \frac{P_T(t)}{Z(x, \lambda)}  \exp\left(-\lambda \sum_{y \in\{-1, 1\}} P_{Y|X}(y| x) \ln \left(\frac{P_{Y|X}(y|x)}{P_{Y|T}(y|t)}\right)\right),
\end{align}    
\end{subequations}
where $Z(x, \lambda)$ is the normalization factor which ensures that $\sum_{t \in \mathcal{T}} P_{T|X}(t|x)$ is equal to $1$.
Note that if $(X,Y)$ is with continuous probability distribution, the BA algorithm is used after the discretization on $X$ and $Y$; thus the resulting distribution $P_{T|X}$ is also discretized. 
However, the BA algorithm does not provide a closed-form solution on the IB problem and its computational complexity is high, in particular for the continuous case. So using the BA algorithm to find the  solution for the IB problem is generally hard. 
In the following section, we will derive several analytically achievable schemes to the problem \eqref{eq:obj_func_1}, and we can identify the performance of our derived solutions by comparing them with the BA algorithm in the simulations in Section~\ref{sec: simulation}.

{
\subsection{State-of-the-art scheme: information dropout method}
\label{sub:ID method}
}

As a state-of-the-art scheme, the information dropout method applies a multiplicative noise as a regularizer to extract essence information under limited capacity \cite{achille2018information}. 
Here, the intermediate representation $T$ takes a structured form, defined as
\begin{align}\label{eq:t_prod}
 T = f_1(X) \odot \eta,   
\end{align}
where $f_1(X)$ is the output of a deep neural network (DNN) with input $X$, and the multiplicative noise $\eta$ follows a log-normal distribution, i.e., $\eta \sim \log \mathcal{N}(0, f_2^2(X))$, with the variance parameter $f_2(X)$ determined by another DNN with input $X$. The parameters of the networks are updated by the optimization problem \eqref{eq:obj_func_2}. In the simulation, the information dropout method is used as a benchmark for comparison.

\section{Achievable bounds for Binary-Gaussian IB problem}
\label{sec: main_results}
With the goal of developing closed-form achievable bounds for \eqref{eq:obj_func_1}, we consider the following generic form
\begin{equation}\label{eq:t}
    T = f_{\text{non-linear}} (X) + N,
\end{equation}
where $f_{\text{non-linear}} \colon \RR \rightarrow \RR$ is a non-linear function, and $N$ is a random variable independent of $X$. 
 Note that  the operation field of the sum in~\eqref{eq:t}  could be real number or binary. 
In the rest of this section, we present achievable bounds for three choices of \eqref{eq:t}, namely, one-bit quantization in Section~\ref{subsec:one-bit}, deterministic quantization in Section~\ref{sec:det_Q}, and soft quantization with $\tanh$ function in Section~\ref{subsec:tanh}.
Under the form of~\eqref{eq:t}, the objective mutual information $I(Y; T)$ writes 
    \begin{align}\label{eq:I_y_t}%
    I(Y; T) &= h(T) - h(T|Y),   \nonumber\\
     &=\! \!-\!\int_{\!-\!\infty}^{\infty} \!\frac{p_{T|Y}(t |1) \!+\! p_{T|Y}(t |\!-\!\!1)}{2} \!\ln\!\frac{p_{T|Y}(t |1) \!+\! p_{T|Y}(t |\!-\!\!1)}{2} \!dt \nonumber \\
    &+ \int_{-\infty}^{\infty} \frac{p_{T|Y}(t |1)}{2}\ln p_{T|Y}(t |1) dt \nonumber \\
    &\quad + \int_{-\infty}^{\infty} \!\frac{p_{T|Y}(t |\!-\!\!1)}{2}\ln\! p_{T|Y}(t|\!-\!\!1 ) dt,
\end{align}
with two conditional probability densities $p_{T|Y}(t|y=1)$ and $p_{T|Y}(t|y=-1)$ given by 
\begin{align}\label{eq:conditional proba}%
    &p_{T|Y}(t |\pm 1) = {\int_{-\infty}^{\infty} p_{X|Y}(x|\pm 1) ~p_{T|X}(t|x) ~ dx} \nonumber \\
    &= \int_{-\infty}^{\infty} \frac{1}{\sqrt{2\pi}} e^{\left(-\frac{(x \mp {\beta})^2}{2} \right)} p_N(t \!-\! f_{\text{non-linear}}(x)) dx,
\end{align}
where $p_N(\cdot)$ denotes the probability density function of the random variable $N$ in \eqref{eq:t}. 

\subsection{An achievable IB solution via two-level random quantization}
\label{subsec:one-bit}
Given a Gaussian mixture observation $X$, 
we first employ the two-level quantization by taking $\overline{X} = f_{\text{non-linear}}(X) =\mathbbm{1}_{X \geq 0}$, where the function is defined in the notation. 
This results in a Markov chain  $Y \rightarrow X \rightarrow \overline{X} \rightarrow T$.  
By the data processing inequality, we have $I(\overline{X}; T) \geq I(X; T)$, 
and therefore a lower (i.e., achievable) bound to the original IB in \eqref{eq:obj_func_1} as
\begin{subequations} \label{eq:IB_HQ_LB}
\begin{align}
    \max_{p_{T| \overline{X}}} \quad &I(Y; T)  \\
    \text{s.t.} \quad &I(\overline{X}; T) \leq R.
\end{align}
\end{subequations}

It is important to note here that both $\overline{X}$ and source $Y$ follow a Bernoulli distribution with equal probability, i.e., $\text{Bern}(1/2)$.  
This scenario is known as doubly symmetric binary sources (DSBS) and has been thoroughly investigated in information theory, see~\cite{goldfeld2020information}. 
Hence, the optimal design is $T=\overline{X} \oplus \overline{N}$, where $\overline{N} \in \{0, 1\}$ follows a Bernoulli distribution with parameter $q$, i.e., $\text{Bern}(q)$. 
This leads to the following result. 

\begin{Proposition}[An achievable IB solution via two-level quantization]\label{prop:one-bit} 
For the IB problem in~\eqref{eq:IB_HQ_LB} with symmetric Bernoulli $Y$ and $X|Y \sim \mathcal N(y \beta, 1)$ as in \eqref{eq:def_model_scalar}, 
then for $0 \leq R \leq \ln 2$, the optimal rate $I^{\star}(Y;T)$ is lower bounded by $I_1(q)$, given by
\begin{align}\label{eq:I_1(q)}
I_1(q) = \ln 2 - H(p(1-q) + q(1-p)), 
\end{align}
where $p = P_{\overline{X}|Y} (\overline{x}= 1|y =-1) = P_{\overline{X}|Y} (\overline{x}= 0|y =1) =   \int_{0}^{\infty} \frac{1}{\sqrt{2 \pi}} \exp(-(x + \beta)^2/2 ) dx$, and where $q$ is the solution to \footnote{$q$ represents the conditional probability $P_{T|\overline{X}}(t=0|\overline{x}=1)$ or $P_{T|\overline{X}}(t=1|\overline{x}=0)$ }
\begin{align}\label{eq:q}
 \ln 2 - H(q) =R,
\end{align}
with $H(q) = -q\ln(q) - (1-q)\ln(1-q)$, and . 
\end{Proposition}
\begin{proof}[Proof of Proposition~\ref{prop:one-bit}]
    See Appendix \ref{sec:proof_one_bit}.
\end{proof}
\noindent
Note that, the IB solution in Proposition~\ref{prop:one-bit} is limited in that it only holds for $0 \leq R \leq \ln 2$; if $R>\ln 2$, $H(q)$ in~\eqref{eq:q} is negative and thus $q$ does not exist.
\begin{remark}[IB solution with two-level quantization for $R\in [0, \ln2)$]\label{rem:one-bit}
When $R=0$ nats, according to the definition of $q$ in \eqref{eq:q}, we have $q = 1/2$, leading to an optimal $I(Y; T)$ of $0$ based on \eqref{eq:I_1(q)}. Similarly, for $R= \ln 2$ nats, the optimal value of $q$ that satisfies \eqref{eq:q} can be either $0$ or $1$. From \eqref{eq:I_1(q)}, we obtain $I(Y; T) = 1 - H(p)$ in this case.
\end{remark}

\subsection{An achievable IB solution via multi-level deterministic quantization}
\label{sec:det_Q} 
In our second approach, we set random noise $N=0$ in \eqref{eq:obj_func_1} and employ an $L$-level deterministic quantizer $\widehat{Q}(\cdot)$ to map the observation $X$ into $L$ bins, with the intermediate representation $T$ given by
\begin{align}
    T &= f_{\text{non-linear}}(X) \overset{\Delta}{=} \widehat{Q}(X).
\end{align}
Here, the quantization points are denoted as ${\{q_i\}}_{i=1}^{L-1}$, with $q_0 = -\infty$ and $q_L = \infty$, and $T$ is quantized as $t_j$ (the center of the quantization region) for $X \in [q_{j-1}, q_j]$, $\forall~ j \in {1, \cdots, L}$.
Consequently, the conditional probability in \eqref{eq:conditional proba} becomes
\begin{align}\label{eq:det_p_t_y}
&\mathbb{P}(T = t_j = \frac{q_{j-1} + q_j}{2}|Y) = \mathbb{P}(q_{j-1} \leq X \leq q_j|Y) \nonumber \\
&= Q(q_{j-1} - \beta Y) - Q(q_j - \beta Y),  \forall j \in {1, \cdots, L},
\end{align}
with $Q(t) = \int_{t}^{\infty} \frac{1}{\sqrt{2\pi}}\exp(-x^2/2) dx$ is the Gaussian Q-function.

Since the mapping from $X$ to $T$ is deterministic, the mutual information $I(X; T)$ becomes the entropy of $T$, i.e., $I(X; T) = H(T)$. We obtain a lower bound to the original IB in \eqref{eq:obj_func_1} by solving the following problem
\begin{subequations} \label{eq:IB_DQ_LB}
\begin{align}
    \max_{\{q_i\}_{i=1}^{L-1}} \quad &I(Y; T)  \\
    \text{s.t.} \quad &H(T) \leq R. \label{eq:constriat_detfunc}
    \end{align}
\end{subequations}

To solve the problem \eqref{eq:IB_DQ_LB} analytically, we can obtain a lower bound by setting the quantization level $L$ as $\lceil e^R \rceil$ and the probability of quantized $T$ space as 
\begin{align}\label{eq:prob_t}
\mathbb{P}(T = t_j) = 
    \begin{cases}
        \frac{1}{\lceil e^R \rceil} - \Delta, & \text{if } j = 1,\\
        \frac{1}{\lceil e^R \rceil} +  \frac{\Delta}{\lceil e^R \rceil-1}, & \text{if } j \neq 1,
    \end{cases}
\end{align}
where the shift value $\Delta$ is determined to satisfy constraint \eqref{eq:constriat_detfunc} as
\begin{align}\label{eq:Delta}
    &H(T) =-\left( \frac{1}{\lceil e^R \rceil} - \Delta\right) \log \left( \frac{1}{\lceil e^R \rceil} - \Delta\right) \nonumber \\
    &- \sum_{j=2}^{L} \left( \frac{1}{\lceil e^R \rceil} +  \frac{\Delta}{\lceil e^R \rceil-1}\right) \log\left( \frac{1}{\lceil e^R \rceil} +  \frac{\Delta}{\lceil e^R \rceil-1}\right) \nonumber \\
    & \overset{\Delta}{=} R.
\end{align}
Therefore, according to \eqref{eq:det_p_t_y}, quantization points $\{q_j\}_{j=1}^{L-1}$ can also be obtained by
\begin{subequations}
    \begin{align}
\mathbb{P}(q_{j-1} \leq X \leq q_j) &= \mathbb{P}(Y=1) \mathbb{P}(q_{j-1} \leq X \leq q_j|Y=1)  \\
   & \quad + \mathbb{P}(Y=-1) \mathbb{P}(q_{j-1} \leq X \leq q_j|Y =-1)  \\
   &= 1/2\left( Q(q_{j-1} - \beta) - Q(q_j - \beta) \right) \nonumber \\
   &\quad +1/2 \left( Q(q_{j-1} + \beta) - Q(q_j + \beta)\right) \nonumber \\
   &\overset{\Delta}{=} \mathbb{P}(T=t_j),
\end{align}
\end{subequations}
where $\mathbb{P}(T=t_j)$ is defined in \eqref{eq:prob_t}.

Note that if $R\leq \ln2$, the quantization level in this scheme is set as $L=2$, similar to the two-level quantization scheme. The deterministic quantization approach outlined above leads to the following proposition.

\begin{Proposition}[An achievable solution to IB via deterministic quantization]\label{prop:derterm_Q} 
For the IB problem in~\eqref{eq:IB_DQ_LB} with symmetric Bernoulli $Y$ and $X|Y \sim \mathcal N(y \beta, 1)$ as in \eqref{eq:def_model_scalar}, then, the optimal rate $I^{\star}(Y; T)$ is lower bounded by $I_2(\Delta)$, the mutual information $I(Y; T)$ given $\Delta$, with $\Delta$ solution to \eqref{eq:Delta}, and the quantization points $\{q_j\}_{j=1}^{\lceil e^R \rceil}$ can be obtained as
\begin{equation}\label{eq:Q_points}
{\mathbb{P}}(q_{j-1} \leq X \leq q_j)=
    \begin{cases}
        \frac{1}{\lceil e^R \rceil} - \Delta  & \text{if } j=1, \\
        \frac{1}{\lceil e^R \rceil} +  \frac{\Delta}{\lceil e^R \rceil-1} & \text{otherwise}.
    \end{cases}
\end{equation}
\end{Proposition}
\noindent
\begin{remark}[IB solution with deterministic quantization for $R \in [0, \infty)$]
For $R=0$ nats, the quantization function $\widehat{Q}(X)$ in Proposition~\ref{prop:derterm_Q} reduces to a single quantization point, resulting in $I(Y; T) = 0$. As $R$ tends to infinity, the quantization becomes finer, ideally leading to $T \approx X$, thereby ensuring that the quantized $T$ closely approximates the observation $X$. In this case, the optimal $I(Y; T)$ converges to $I(X; Y)$.
\end{remark}

\subsection{An achievable IB solution via soft quantization}
\label{subsec:tanh}
Here, we propose to solve the IB problem by \emph{jointly} tuning the non-linear function \emph{and} the noise $N$.
We first use the hyperbolic tangent $\tanh$ function to the observations $X$, which can be viewed as a ``soft'' quantization to obtain the value between $-1$ and $1$, instead of binary values $\pm 1$, from the mixture Gaussian observation $X$. The core idea of applying $\tanh$ function is inspired from that the Minimum Mean Square Error
(MMSE) estimation of the binary source $Y$ given the Gaussian mixture $X$ is
$\tanh{(\beta X)}$ \cite{guo2005mutual}. 
After the $\tanh$ non-linearity, Gaussian noise is then added to the intermediate representation $T$ as
\begin{align}
    T &= f_{\text{non-linear}}(X) + \widetilde{N} \nonumber \\
    &= \tanh(\beta X) + \widetilde{N}, 
\end{align} 
with $\widetilde{N} \sim \mathcal{N}(0, \alpha^{-2})$. 
In terms of mutual information $I(X; T)$ or $I(Y; T)$, this is equivalent to
\begin{align}\label{eq:t_MMSE}
        T =  \alpha \tanh(\beta X) + \widehat{N}, 
\end{align}
with $\widehat{N} \sim \mathcal{N}(0, 1)$, and let $\widehat{X} \overset{\Delta}{=} \tanh{\beta X} $. Since the $\tanh$ function is a one-to-one mapping, we have $I(\widehat{X}; T)=I(X; T)$ and thus  the IB problem becomes
\begin{subequations}
\begin{align}
    \max_{\alpha \geq 0} \quad & I(Y; T)  \\
    \textrm{s.t.} \quad &I(\widehat{X}; T) \leq R, \label{eq:variation IB constr 1}\\
    &T|\widehat{X} \sim \mathcal{N} \left(\alpha \widehat{X}, 1\right),  \label{eq:variation IB constr 2}
\end{align}
 \label{eq:V}%
\end{subequations}
where $I(\widehat{X}; T)$ can be computed as follows, 
\begin{align}\label{eq:I_Xhat_T}
    I(\widehat{X}; T) &= h(T) - h(T| \widehat{X}) \nonumber \\
    &=-\int p_T(t) \ln(p_T(t)) dt -\frac{1}{2} \ln( 2\pi e).
\end{align}
Since it is still complicated to compute
$\alpha$ in closed-form satisfying $I(\widehat{X}; T) \overset{\Delta}{=} R$. We further derive a lower bound on $-\int p_T(t) \ln(p_T(t)) dt$ 
 by introducing a variational distribution of $T$ (denoted by $q_T(\cdot)$) and 
 by using the information inequality~\cite[Theorem 2.6.3]{hanineq}, we have 
\begin{align}
     I(\widehat{X}; T)
    &\leq -\int p_T(t) \ln(q_T(t)) dt -\frac{1}{2} \ln( 2\pi e).\label{eq:upper bound of I hat based on q}
\end{align}
Then we need to find out a reasonable  variational distribution $q_T(\cdot)$. 
Since $\widehat{X}$ is the MMSE estimation of $Y$, we can design the variational distribution of $\widehat{X}$ as Bernoulli distribution, i.e., $q_{\widehat{X}}(\widehat{X}=-1) = q_{\widehat{X}}(\widehat{X}=1) = \frac{1}{2}$ to simplify the computation of $\ln q_T(t)$. Intuitively speaking, the less the noise power of $X$ is, the closer the variational distribution $q_{\widehat{X}}$ gets to the true distribution $p_{\widehat{X}}$. Hence, the variational distribution of $T$ is given by
\begin{subequations}
   \begin{align}
    q_T(t) &= \int_{-1}^1 p_{T|\widehat{X}}(t|\widehat{x}) q_{\widehat{X}}(\widehat{x}) d\widehat{x} \\
    &= \frac{1}{\sqrt{2 \pi}} \exp(-\frac{t^2 + {\alpha}^2}{2}) (\cosh{(\alpha t)}). \label{eq:qt}
\end{align} 
\end{subequations}
To simplify notations, we denote 
\begin{subequations}
    \begin{align}
    f(\beta) &\overset{\Delta}{=}  \int_{-1}^1 p_{\widehat{X}} (\widehat{x}) \widehat{x}^2 d \widehat{x}, \\
    g(\beta)  &\overset{\Delta}{=} \int_{-1}^1 p_{\widehat{X}} (\widehat{x}) |\widehat{x}| d \widehat{x} = 2\int_{-1}^0 p_{\widehat{X}} (\widehat{x}) (-\widehat{x}) d \widehat{x}, \label{eq: g_beta_2}
\end{align}\label{eq:f_g_beta}
\end{subequations}
where \eqref{eq: g_beta_2} holds since $p_{\widehat{X}} (\widehat{x})$ in \eqref{eq:p_yhat} is an even function.

By taking \eqref{eq:qt} into \eqref{eq:upper bound of I hat based on q}, an upper bound to $I(\widehat{X}; T)$ based on variational distribution is derived as 
\begin{align}
    I(\widehat{X}; T) &\leq \frac{{\alpha}^2}{2}(1 + f(\beta) ) \underbrace{-\int_{-\infty}^{\infty} \left(\int_{-1}^{1} p_{T|\widehat{X}} (t |\widehat{x}) p_{\widehat{X}} (\widehat{x}) d\widehat{x}\right) \ln(\cosh(\alpha t)) dt}_{(d)}.\label{eq:UB}
\end{align}
Next we propose two upper bounds on \eqref{eq:UB} by deriving lower bounds on $\ln(\cosh{\alpha t})$: 
    \begin{itemize}
        \item[(i)]  The first bound  is based on the inequality   $\ln(\cosh(x)) \geq \sqrt{1 + x^2} -1$, and hence an upper bound to $(d)$ in \eqref{eq:UB} is derived as 
\begin{subequations}
    \begin{align}
    &\!-\!\int_{-1}^1 p_{\widehat{X}} (\widehat{x}) \! \int_{\!-\infty}^{\infty} \! \frac{1}{\sqrt{2\pi}} \exp(-\frac{(t\!-\!\alpha \widehat{x})^2}{2})\ln(\cosh(\alpha t)) dt d\widehat{x}  \\
    &\leq \!-\!\int_{-1}^1 p_{\widehat{X}} (\widehat{x}) \int_{-\infty}^{\infty} \frac{1}{\sqrt{2\pi}} \exp(-\frac{(t-\alpha \widehat{x})^2}{2})  \left[\sqrt{1 + \alpha^2t^2} -1\right] dt d\widehat{x} ,  \nonumber \\
    &\leq \!-\!\int_{-1}^1 p_{\widehat{X}} (\widehat{x}) \left[\sqrt{1 + \alpha^4 \widehat{x}^2}\right] d\widehat{x} + 1 \label{eq:by convexity of f} \\
    &= -\int_{-1}^0 2p_{\widehat{X}} (\widehat{x}) \left[\sqrt{1 + \alpha^4 \widehat{x}^2}\right] d\widehat{x}  + 1 ,\label{eq:LB_1_part1}
\end{align}
\end{subequations}
where~\eqref{eq:by convexity of f} comes from  the convexity of function $f(t) = \sqrt{1 + \alpha^2 t^2}$, i.e., $\mathbb{E}\left[\sqrt{1 + \alpha^2 t^2}\right] \geq \sqrt{1 + \alpha^2 (\mathbb{E}[t])^2}$, and~\eqref{eq:LB_1_part1} follows  
  since $p_{\widehat{X}} (\widehat{x})$ and $\sqrt{1 + \alpha^4 {\widehat{x}}^2}$ are both even functions regarding to $\widehat{x}$.
Based on the Jensen's inequality, $\int_{-1}^0 2p_{\widehat{X}} (\widehat{x}) d \widehat{x} = 1$, and notation for $g(\beta)$, an upper bound of the RHS of~\eqref{eq:LB_1_part1} is given by
\begin{subequations}
\begin{align}
    -\int_{-1}^0 2p_{\widehat{X}} (\widehat{x}) \left[\sqrt{1 + {\alpha}^4 \widehat{x}^2}\right] d\widehat{x}  + 1 &\leq \sqrt{1 + {\alpha}^4  \left(\int_{-1}^0 2 p_{\widehat{X}} (\widehat{x}) \widehat{x} d \widehat{x}\right)^2} + 1   \\
    &= -\sqrt{1 + {\alpha}^4  \left(g(\beta)\right)^2} + 1.\label{eq:UB_1}
\end{align}
\end{subequations}
By taking~\eqref{eq:UB_1} and~\eqref{eq:LB_1_part1} into~\eqref{eq:UB}, we obtain the following upper bound of $I(\widehat{X}; T)$,
\begin{align}
    I(\widehat{X}; T) \leq \frac{{\alpha}^2}{2} (1 + f(\beta)) -\sqrt{1 + {\alpha}^4  (g(\beta))^2} + 1,
    \label{eq:first upper bound of Ixt}
\end{align}

\item[(ii)]  The second bound is based on $\ln(\cosh(x)) \geq x - \ln2, ~\forall~ x \geq 0$, which is tighter than the first lower bound on $\ln(\cosh{x})$ for relatively large $x$, resulting in a tighter upper bound on $I(\widehat{X};  T)$. However, the second bound only holds for $R \geq \ln2$. Hence, by separating the negative part and positive part of $t$ and introducing an auxiliary variable $s$ defined as $ s = t - \alpha \widehat{x}$,  the second upper bound to $(d)$ in \eqref{eq:UB} is derived as
\begin{subequations}
\begin{align}
    &-\int_{-\infty}^{\infty} \left(\int_{-1}^{1} p_{T|\widehat{X}} (t |\widehat{x}) p_{\widehat{X}} (\widehat{x}) d\widehat{x}\right) \ln(\cosh(\alpha t)) dt  \nonumber \\
    &\leq -  \int_{-1}^{1} p_{\widehat{X}} (\widehat{x}) \left( \int_{-\infty}^{0} p_{T|\widehat{X}} (t |\widehat{x}) \left[-\alpha t - \ln2\right] dt \right) d\widehat{x} \nonumber \\
    & \quad - \int_{-1}^{1} p_{\widehat{X}} (\widehat{x}) \left( \int_{0}^{\infty} p_{T|\widehat{X}} (t |\widehat{x}) \left[\alpha t - \ln2\right] dt \right) d\widehat{x}, \label{eq:LB_2_original_bound} \\
    &=-\int_{-1}^{1} 
    p_{\widehat{X}} (\widehat{x}) \left( \int_{-\infty}^{-\alpha \widehat{x}} \frac{1}{\sqrt{2 \pi}} \exp(-\frac{s^2}{2}) \left[-\alpha (s + \alpha \widehat{x})\right] ds \right) d\widehat{x} \nonumber \\
    &\quad -\int_{-1}^{1} p_{\widehat{X}} (\widehat{x}) \left( \int_{-\alpha \widehat{x}}^{\infty}  \frac{1}{\sqrt{2 \pi}} \exp(-\frac{s^2}{2}) \left[\alpha (s + \alpha \widehat{x})\right] ds \right) d\widehat{x} + \ln2. \label{eq:LB_2_derivation}
\end{align}
\end{subequations}
Moreover, using that fact that $\int (-s) \exp(-\frac{s^2}{2})  ds = \exp(-\frac{s^2}{2})$, and separating the negative part and positive part of $\widehat{x}$, \eqref{eq:LB_2_derivation} is further developed as 
\begin{subequations}
\begin{align}
    & \ln2 - \frac{2 \alpha}{\sqrt{2\pi}}\int_{-1}^{1} p_{\widehat{X}} (\widehat{x}) \exp(-\frac{{\alpha}^2 \widehat{x}^2}{2}) d\widehat{x} \nonumber \\
    &\quad - \int_{-1}^{1} p_{\widehat{X}} (\widehat{x}) \left[-\alpha^2  \widehat{x}\right]  \int_{-\infty}^{-\alpha \widehat{x}} \frac{1}{\sqrt{2 \pi}} \exp(-\frac{s^2}{2}) ds  d\widehat{x}  \nonumber \\
    &\quad - \int_{-1}^{1} p_{\widehat{X}} (\widehat{x}) \left[\alpha^2  \widehat{x}\right] \int_{-\alpha \widehat{x}}^{\infty} \frac{1}{\sqrt{2 \pi}} \exp(-\frac{s^2}{2}) ds d\widehat{x}  \nonumber \\
    &= \ln2 - \frac{2 \alpha}{\sqrt{2\pi}}\int_{-1}^{1} p_{\widehat{X}} (\widehat{x}) \exp(-\frac{{\alpha}^2 \widehat{x}^2}{2}) d\widehat{x} \nonumber \\
    &\quad + \frac{\alpha^2  }{\sqrt{2 \pi}} \int_{-1}^{0} \widehat{x} p_{\widehat{X}} (\widehat{x})  \left[\int_{\alpha \widehat{x}}^{-\alpha \widehat{x}}  \exp(-\frac{s^2}{2}) ds \right]  d\widehat{x} \nonumber \\
    &\quad + \frac{\alpha^2  }{\sqrt{2 \pi}} \int_{0}^{1} \widehat{x} p_{\widehat{X}} (\widehat{x})  \left[-\int_{-\alpha \widehat{x}}^{\alpha \widehat{x}}  \exp(-\frac{s^2}{2}) ds \right]  d\widehat{x}  \\
    &= \ln2 - \frac{2 \alpha}{\sqrt{2\pi}}\int_{-1}^{1} p_{\widehat{X}} (\widehat{x}) \exp(-\frac{{\alpha}^2 \widehat{x}^2}{2}) d\widehat{x} \nonumber \\
    &\quad +  \underbrace{2\alpha^2  \int_{-1}^{0} \widehat{x} p_{\widehat{X}} (\widehat{x})  \left[\int_{\alpha \widehat{x}}^{-\alpha \widehat{x}} \frac{1}{\sqrt{2 \pi}} \exp(-\frac{s^2}{2}) ds \right]  d\widehat{x}}_{(f)}.\label{eq:LB_2_derivation_2}
\end{align}
\end{subequations}
For any non-negative real number $\alpha$ and negative real number $\widehat{x} < 0$, an upper bound to the Gaussian $Q$ function $Q(-\alpha \widehat{x})$ is derived as 
\begin{subequations}
   \begin{align}
    Q(-\alpha \widehat{x}) &= \int_{-\alpha \widehat{x}}^{\infty} \frac{1}{\sqrt{2\pi}} \exp(-\frac{s^2}{2}) ds  \\
    &\leq \int_{-\alpha \widehat{x}}^{\infty} \frac{1}{\sqrt{2\pi}} \frac{s}{-\alpha\widehat{x}} \exp(-\frac{s^2}{2}) ds \label{eq:x_exp_a} \\
    &= \frac{1}{\alpha \widehat{x} \sqrt{2 \pi}} \left(- \exp\left(-\frac{\alpha^2 \widehat{x}^2}{2}\right)\right),
\end{align}\label{eq:x_exp_bound} 
\end{subequations}
where \eqref{eq:x_exp_a} holds since $\frac{s}{-\alpha\widehat{x}}$ is always larger than $1$ in the integral region.
Therefore, also note that $\widehat{x}$ in the $(f)$ of \eqref{eq:LB_2_derivation_2} in the integral region is always non-positive, based on the inequality \eqref{eq:x_exp_bound}, we can derive an upper bound on $(f)$ in \eqref{eq:LB_2_derivation_2} as  
\begin{subequations}
    \begin{align}
    (f) &= 2{\alpha}^2 \int_{-1}^{0} p_{\widehat{X}} (\widehat{x}) \widehat{x} \left[1 - 2Q(-a \widehat{x})\right] d \widehat{x}, \\
    &\leq 2\alpha^2 \int_{-1}^{0} p_{\widehat{X}} (\widehat{x}) \widehat{x} \left[1 - \frac{2}{\alpha\widehat{x} \sqrt{2 \pi}} \left(- \exp\left(-\frac{\alpha^2 \widehat{x}^2}{2}\right)\right)\right] d \widehat{x} \label{eq:random_gaussian_ineq}
    \end{align}
\end{subequations}
Hence, by taking \eqref{eq:random_gaussian_ineq} into \eqref{eq:LB_2_derivation_2} and combining \eqref{eq:UB}, we can further relax the constraint and obtain the following upper bound on $I(\widehat{X}; T)$ 
\begin{subequations}
    \begin{align}
    I(\widehat{X}; T) &\leq \ln2 + 2\alpha^2 \int_{-1}^{0} p_{\widehat{X}} (\widehat{x}) \widehat{x} d \widehat{x} + \frac{\alpha^2}{2}(1 + f(\beta)) \\
    &= \alpha^2\left[\frac{1}{2} + \frac{f(\beta)}{2} -g(\beta)\right]  + \ln2. \label{eq:UB_2_constraint_UB}
\end{align}
\end{subequations}
    \end{itemize} 
    Next, we solve $\alpha$  analytically satisfying that $R$ is equal to each upper bound of $ I(\widehat{X}; T)$ in the RHS of \eqref{eq:first upper bound of Ixt} and \eqref{eq:UB_2_constraint_UB}, and the obtained solution is also an achievable solution for the IB problem in~\eqref{eq:variation IB constr 2}. 
Finally, the value of the mutual information $I(Y; T)$ is obtained for the corresponding value of $\alpha$. 
The above is the intuitive proof  for the following result, whose detailed proof is given in Appendix~\ref{sec:proof_main}.
\begin{Proposition}[An achievable solution to IB via soft quantization]\label{theo:main-results} 
For the IB problem defined in~\eqref{eq:V} with symmetric Bernoulli $Y$ and $X|Y \sim \mathcal N(\beta Y, 1)$, the optimal rate $I^{\star}(Y; T)$ is lower bounded by 
$\max\{I_3(\alpha_{\text{lb}_1}) , I_4(\alpha_{\text{lb}_2}) \}$
if $R\geq \ln 2$, and lower bounded by $I_3(\alpha_{\text{lb}_1})$ otherwise, with
\begin{subequations}
    \begin{align}
    &\alpha_{\text{lb}_1} =\! {\sqrt{\!\frac{ (R\!-\!1)(1 \!\!+\!\! f(\beta)) \!+\!\! \sqrt{((1 \!\!+ \!\!f(\beta))^2  \!\!+ \!4g\!^2(\beta)(R^2 \!\!-\!\! 2R))}}{((1 \!+\! f(\beta))^2 \! -\! 4 g^2(\beta))/2}}}, \label{eq:LB_1}\\
    &\alpha_{\text{lb}_2} \!=\! \sqrt{\frac{R - \ln2}{\frac{1}{2} + \frac{f(\beta)}{2} - g(\beta)}}, ~~~\text{if}~ R \geq \ln2,  \label{eq:LB_2}
\end{align}\label{eq:alpha}%
\end{subequations}
where, for the ease of presentation, we define $I_3(\alpha_{\text{lb1}}))$ and $I_4(\alpha_{\text{lb2}}))$ as the mutual information $I(Y; T)$ given $\alpha_{\text{lb1}}$ and $\alpha_{\text{lb2}}$ respectively, $\widehat{X}:=\tanh{(\beta X)}$, and\footnote{\label{foot:const}Note that $f(\beta)$ and $g(\beta)$ are deterministic functions of $\beta$.} $f(\beta)$ and $g(\beta)$ are defined in \eqref{eq:f_g_beta}. 
\end{Proposition}

\begin{remark}[IB solution with soft quantization for $R \in [0, \infty)$]\label{rem:limiting_cases_soft}
First, for $R=0$, $\alpha_{\text{lb}1}= 0$ according to its definition in \eqref{eq:LB_1}, which means that $I(Y; T) = 0$.
Next, as $R \rightarrow \infty$, both $\alpha_{\text{lb}1}$ and $\alpha_{\text{lb}_2}$ tend to infinity according to \eqref{eq:alpha}.
With the intermediate representation design in \eqref{eq:t_MMSE}, as $\alpha \rightarrow \infty$, $T$ converges to $\widehat{X}$, allowing the objective mutual information $I(Y; T)$ to approach the optimal value $I(X; Y)$. These results are confirmed by simulations in the appendix~\ref{ch:extreme_points}.
\end{remark}

\subsection{A unified achievable scheme to IB}
By combining the three proposed achievable schemes in Proposition~\ref{prop:one-bit}--\ref{theo:main-results}, 
we obtain the analytic achievable scheme to IB in Theorem~\ref{theo:main-results-2} as follows.

\begin{Theorem}[An analytic and achievable scheme to IB under Gaussian mixtures]\label{theo:main-results-2} 
For the IB problem in \eqref{eq:obj_func_1}, the optimal rate $I^*(Y; T)$ is lower bounded by $\max\{I_1(q), I_2(\Delta), I_3(\alpha_{\text{lb1}}), I_4(\alpha_{\text{lb2}})\}$, for $I_1(q)$, $I_2(\Delta), I_3(\alpha_{\text{lb1}})$, $I_4(\alpha_{\text{lb2}})$ defined in Proposition~\ref{prop:one-bit}--\ref{theo:main-results}.
\end{Theorem}

\begin{remark}[Extension to QPSK setting]
    Our proposed achievable schemes can be easily extended to the case of i.i.d.\@ output from the Quadrature Phase Shift Keying (QPSK). These sequences can be viewed as two parallel sets of i.i.d.\@ sequences of BPSK. Our proposed achievable schemes can be effectively applied to each of these sequences to address the IB problem.
\end{remark}

\section{{Extension to vector mixture Gaussian problem}}
\label{sec: application}

In this section, we extend the achievable analytic IB scheme proposed in Section~\ref{sec: main_results} to multivariate mixture Gaussian model.  
For label $Y$ drawn from  a symmetric Bernoulli distribution (that is, $Y = \pm 1$ with $P(Y=-1) = P(Y=1) = 1/2$), the data vector $\x=(x_1,\ldots,x_{d_0}) \in\RR^{d_0}$ follows a GMM and depends on the label $Y$ in such as way that
\begin{equation}\label{eq:def_model1}
  \x = \bbeta \cdot Y + \boldsymbol{\epsilon},
\end{equation}
for some \emph{deterministic} vector $\bbeta=[\beta_1,\ldots,\beta_{d_0}]^{\rm T} \in \RR^{d_0}$ and Gaussian random noise $\boldsymbol{\epsilon}=[\epsilon_1,\ldots,\epsilon_{d_0}]^{\rm T} \sim \mathcal N (\mathbf{0}, \I_{d_0})$.
In the context of IB, we are interested in constructing an intermediate representation $\bt=[t_1,\ldots,t_d]^{\rm T} \in \RR^d$ of $\x$ to solve the IB problem 
 \begin{subequations}
     \begin{align}
    \max_{p(\bt|\x)} \quad & I(Y; \bt)  \\
    \textrm{s.t.} \quad &I(\x; \bt) \leq R, \label{eq:security constr} 
\end{align}\label{eq:obj_func_11}%
\end{subequations}
for some given $R \geq 0$. 
Here we focus on the setting of $d = d_0$ and, 
for each $i\in \{ 1,\ldots, d_0 \}$, optimize the conditional distribution $p(t_i|x_i)$ by solving the {following IB problem}, 
\begin{subequations}
\begin{align}
    \max_{\{p(t_i|x_i)\}_{i=1}^{d_0}} \quad & I(Y; \bt)  \\
    \textrm{s.t.} \quad~~~~ &I(x_i; t_i) \leq R_i, ~\forall i \in \{1, \cdots, d_0\}\label{eq:variation IB constr 111}
\end{align}
\label{eq:obj_func_entry}%
\end{subequations}
for some $R_i\geq 0$ such that
\begin{align}
    R_1+\cdots+R_{d_0}=R. \label{eq:rate allocation}
\end{align}
{In Appendix~\ref{sec:achievability proof}, we prove that any achievable solution of~\eqref{eq:obj_func_entry}  
is also an achievable solution of the problem in~\eqref{eq:obj_func_11}; i.e., any $\{p(t_i|x_i):i \in \{ 1,\ldots, d_0 \} \}$ satisfying the constraints~\eqref{eq:variation IB constr 111} also leads a distribution $p(\bt|\x)$  satisfying the constraint in~\eqref{eq:security constr}}.

\section{Application to scalar Gaussian mixture classification }
\label{sec:connection_ML}

The IB problem for Gaussian mixture observations has direct implications for the fundamental problem of binary GMM classification with information leakage, where mutual information serves as the privacy metric. 
With the IB framework, we can extract a maximally compressed yet informative feature for the GMM classification task. 
The misclassification error rate, based on the design of the intermediate representation $T$ in \eqref{eq:t}, is given by:
\begin{align}
    {\rm Pr}(Y \neq \widehat{Y}) &= \frac{1}{2} {\rm Pr} (\widehat{Y}=1|Y =-1) + \frac{1}{2} {\rm Pr} (\widehat{Y}=-1|Y =1),
\end{align}
where $\widehat{Y}$ is the estimate of $Y$ based on the intermediate representation $T$.
In the following, the misclassification error rates of the three achievable schemes are given, providing the fundamental tradeoff between GMM classification performance and the information leakage $I(X;T)$ under the IB formulation.

\subsection{Two-level random quantization scheme}

\begin{Proposition}[Classification error via two-level quantization]\label{prop:one-bit_2} 
For the IB problem in~\eqref{eq:IB_HQ_LB} with symmetric Bernoulli $Y$ and $X|Y \sim \mathcal N(y \beta, 1)$ as in \eqref{eq:def_model_scalar}, based on the formulated  Markov chain, $Y \rightarrow \overline{X} \rightarrow T$, where $\overline{X} = \mathbbm{1}_{X \geq 0}$ and $T = \overline{X} \oplus \overline{N}$, and given the estimator as 
\begin{equation}
\widehat{Y}=
    \begin{cases}
        1 & \text{if } T = 1,\\
        -1 & \text{if } T = 0,
    \end{cases}
\end{equation}
the misclassification error rate of the this scheme is given by
    \begin{align}
         {\rm Pr}(Y \neq \widehat{Y}) = (1- p)q  + p (1- q).
    \end{align}
    Using $I^*(q)$ in \eqref{eq:I_1(q)}, we have $I^*(q) = \ln 2 - H({\rm Pr}(\widehat{Y} \neq Y))$. 
\end{Proposition}
\begin{proof}[Proof of Proposition~\ref{prop:one-bit_2}]
    See Appendix \ref{sec:proof_one_bit_2}.
\end{proof}

\subsection{Multi-level deterministic quantization scheme}

\begin{Proposition}[Classification error to IB via deterministic quantization]\label{prop:derterm_Q_2} 
For the IB problem in~\eqref{eq:IB_DQ_LB} with symmetric Bernoulli $Y$ and $X|Y \sim \mathcal N(y \beta, 1)$ as in \eqref{eq:def_model_scalar}, based on the Markov chain $Y \rightarrow {X} \rightarrow T$, where $T = \widehat{Q}(X)$, and given the estimator as 
\begin{equation}
\widehat{Y}=
    \begin{cases}
        1 & \text{if } T \geq 0,\\
        -1 & \text{if } T < 0,
    \end{cases}
\end{equation}
the misclassification error rate of the multi-level deterministic quantization is given by
\begin{align}
     {\rm Pr}(Y \neq \widehat{Y}) &= \frac12 \left(Q(-q_s + \beta) + Q(q_s + \beta)\right),
\end{align}
where assuming that the quantization points for $T$ are $t_1 \leq t_2  \cdots \leq t_{L}$, the index $s$ indicates the subscript of the quantization point which satisfies $t_s <0$ and $t_{s+1} \geq 0$.
\end{Proposition}
\begin{proof}[Proof of Proposition~\ref{prop:derterm_Q_2}]
    See Appendix \ref{sec:proof_determ_Q_2}.
\end{proof}

\subsection{Soft quantization scheme}

\begin{Proposition}[Classification error via soft quantization]\label{prop:soft_quantization} 
For the IB problem defined in~\eqref{eq:V} with symmetric Bernoulli $Y$ and $X|Y \sim \mathcal N(\beta Y, 1)$, based on the Markov chain, $Y \rightarrow {X} \rightarrow T$, where $T = \alpha \widehat{X} + \widehat{N} =\alpha \tanh{(\beta X)} + \widehat{N}$, and given the estimator as 
\begin{equation}
\widehat{Y}=
    \begin{cases}
        1 & \text{if } T \geq 0,\\
        -1 & \text{if } T < 0,
    \end{cases}
\end{equation}
the misclassification error rate of the this scheme is given by
    \begin{align}
        {\rm Pr}(Y \neq \widehat{Y}) = \frac{1}{2 \pi} \int_{0}^{\infty} \int_{-\infty}^{\infty} e^{-\frac{t^2 + \alpha^2 \tanh^2{(\beta x)} + x^2 +\beta^2}{2}} \cosh(t \alpha \tanh{(\beta x)} - \beta x)  dxdt. 
    \end{align}
\end{Proposition}
\begin{proof}[Proof of Proposition~\ref{prop:soft_quantization}]
    See Appendix \ref{sec:proof_soft_quantization}.
\end{proof}

\section{Simulation Results}
\label{sec: simulation}
\subsection{Evaluation of the BA algorithm}
In this section, we present three baseline iterative algorithms for evaluation.
\subsubsection{Three baseline algorithms on the IB problem}
\paragraph{{Agglomerative Information Bottleneck (Agg-IB)}}
\label{sec: agg_IB}
Inspired by the iterative algorithm in Section \ref{sec: BA}, this algorithm aims to introduce a hard partition on the observation $X$ into $m$ disjoint subsets to maximize the objective function in \eqref{eq:obj_func_1} \cite{slonim1999Agglomerative}.
For notional simplicity, we define $T_{\ell}$ as the merged space of $T$ based on $\ell$ partitions. 
First, we discretize the space of $X$ into $d_{X}$ clusters, and duplicate the discrete space $X$ as the $T$ space, i.e., $X, T_{d_X} \in \{t_1, t_2, ..., t_{d_X}\}$, leading to $I(T_{d_X}; Y) = I(X; Y)$.
Furthermore, we reduce the cardinality of $T$ by iteratively merging the two clusters of $T$ in such a way that the objective function is maximized until the desired number of subsets $m$ is reached. Thus, the iteration forms a Markov chain, $T_{d_X} \rightarrow T_{d_X-1} \rightarrow \cdots \rightarrow T_{m}$. 
The selection of two clusters to merge into $\ell$ subsets depends on the difference of the objection function, denoted as $\Delta L(\cdot, \cdot)$. Considering merging two clusters $t_i, t_j$ with the probabilities $P_T(t_i), P_T(t_j)$ respectively, the difference of the objection function
can be defined as 
\begin{align}\label{eq:delta_L}
    \Delta L (t_i, t_j) &=  I(T_{\ell+1}; Y) - I(T_{\ell}; Y)\nonumber \\
    &= (P_T(t_i) + P_T(t_j))  D_{\rm JS}^{\Pi}(P_{Y|T}(y|t_i) \| P_{Y|T}(y|t_j)), 
\end{align}
where $D_{\rm JS}^{\Pi}(\cdot\|\cdot)$ denotes the Jensen-Shannon divergence. The indices of the merging clusters can be determined by 
\begin{align}
  ({\rm idx}_i, {\rm idx}_j) = \arg \min_{i, j \in [\ell +1], i \neq j}   \Delta L (t_i, t_j),
\end{align}
ensuring that the objective function $ I(T_{\ell}; Y)$ is maximized when $t_{{\rm idx}_i}$ and $t_{{\rm idx}_j}$ emerge from all available fusion possibilities at this iteration.

\paragraph{Sequential Information Bottleneck (Seq-IB)} 
\label{sec:Seq-IB}
The Seq-IB algorithm is a response to resolving the computational complexity issue in the Agg-IB algorithm \cite{Hassanpour2017overview}. Instead of duplicating the space of $X$ as the space of $T$, the Seq-IB algorithm initializes with a random partition of $X$ with $m$ clusters forming the space of $T$, i.e., $T \in \{t_1, t_2, ...,t_m\}$. At each iteration, a new point $x^{\rm new}$ distinct from the cluster points is randomly drawn as a new cluster. The agglomerative clustering algorithm detailed in Section \ref{sec: agg_IB} is then employed to merge this new cluster into the existing clusters, maximizing the objective function $ I(T_{\ell}; Y)$ \cite{slonim2002unsupervised}. The merging decision is determined by 
\begin{align}
    t^{\rm new} = \arg \min_{t \in \{t_1, ..., t_m\}} \Delta L(t, x^{\rm new}),
\end{align}
where $\Delta L(\cdot, \cdot)$ is defined in \eqref{eq:delta_L}. The probability of the new cluster point is then updated as the sum of the probabilities of the two merged clusters. This iterative process continues until the convergence criterion is met. To mitigate the risk of converging to local minima, \cite{slonim2002unsupervised} recommends running the algorithm with various initializations.  

\paragraph{Deterministic Information Bottleneck (Det-IB)}
\label{sec:det_IB}
The Det-IB algorithm is inspired by the solution of the generalized IB problem as
\begin{align}\label{eq:generalized_IB}
    \widetilde{L} = \min_{f_{T|X}(t|x)} H(T) - \gamma H(T|X) - \lambda I(T; Y), 
\end{align}
where $\gamma \in [0,1]$. In some special cases, for instance, when $\gamma=1$, it aligns with the original problem formulated in \eqref{eq:obj_func_2}, while $\gamma=0$ corresponds to the deterministic quantization scheme in \eqref{eq:IB_DQ_LB}. Using the Blahut-Arimoto algorithm to address \eqref{eq:generalized_IB}, it iterates over probabilities as described below \cite{strouse2017deterministic}
\begin{align}
        P^{\gamma}_{T|X}(t|x) &= \frac{1}{Z(x, \gamma, \lambda)} \exp\left(\frac{1}{\gamma} \left( \log  P^{\gamma}_{T}(t) - \lambda D_{\rm KL} (P_{Y|X}(y|x)\| P^{\gamma}_{Y|T} (y|t))\right)\right), \label{eq:det_IB_a}\\
        P^{\gamma}_{T}(t) &= \sum_{x \in \mathcal{X}}  P^{\gamma}_{T|X}(t|x)  P_X(x), \label{eq:det_IB_b}\\
    P^{\gamma}_{Y|T}(y|t) &= \frac{1}{P^{\gamma}_T(t)} \sum_{x \in \mathcal{X}} P_{Y|X}(y|x)  P_X(x) P^{\gamma}_{T|X}(t|x), \label{eq:det_IB_c}
\end{align}
where $Z(x, \gamma, \lambda)$ denotes a normalization factor ensuring $\sum_{t \in \mathcal{T}} P^{\gamma}_{T|X}(t|x)$ equals $1$.
The Det-IB algorithm aims to solve the problem \eqref{eq:generalized_IB} specifically for $\gamma = 0$. This simplifies \eqref{eq:det_IB_a} as follows
\begin{align}\label{eq:det_IB_aa}
    \lim_{\gamma \rightarrow 0} P^{\gamma}_{T|X}(t|x) = \delta\left(\arg \max_t  \left(\log P^{\gamma}_T(t) - \lambda D_{\rm KL} (P_{Y|X}(y|x)\| P^{\gamma}_{Y|T} (y|t)))\right)\right),
\end{align}
where $\delta(\cdot)$ is defined as the Dirac delta distribution. The Det-IB algorithm begins with a random deterministic quantization $P^{\gamma}_{T|X}(t|x)$ and iterates through the equations \eqref{eq:det_IB_b}, \eqref{eq:det_IB_c} and \eqref{eq:det_IB_aa} until the convergence criterion is satisfied.

\subsubsection{Simulation on the evaluation of the BA algorithm}
\begin{figure}[htb]
\centering
        \begin{subfigure}[b]{0.32\textwidth}
        \begin{tikzpicture}[scale=.5]
        \pgfplotsset{every major grid/.append style={densely dashed}}
    \begin{axis}[%
    width=3.528in,
    height=3.029in,
    at={(1.011in,0.767in)},
    scale only axis,
    xmin=0,
    xmax=2.5,
    xlabel style={font=\color{white!15!black}},
    xlabel={$I(X; T)$ (bits)},
    ymin=0,
    ymax=0.22,
    ylabel style={font=\color{white!15!black}},
    ylabel={$I(Y; T)$ (bits)},
    axis background/.style={fill=white},
    xmajorgrids,
    ymajorgrids,
    legend style={at={(0.64,0.0)}, anchor=south west, legend cell align=left, align=left, draw=white!15!black}
    ]
    \addplot [color=mycolor1, line width=2.5pt, mark size=3.0pt, mark=+, mark options={solid, mycolor1}]
      table[row sep=crcr]{%
        0.0000    0.0000\\
        0.1149    0.0301\\
        0.3405    0.0810\\
        0.8077    0.1482\\
        1.0976    0.1720\\
        1.3077    0.1842\\
        1.4732    0.1916\\
        1.6078    0.1966\\
        1.7214    0.2001\\
        1.8202    0.2027\\
        1.9070    0.2048\\
        1.9848    0.2064\\
        2.0550    0.2077\\
        2.1191    0.2089\\
        2.1780    0.2098\\
        2.2325    0.2106\\
        2.2832    0.2113\\
        2.3306    0.2119\\
        2.3751    0.2125\\
        2.4170    0.2129\\
        2.4566    0.2134\\
        2.4942    0.2137\\
    };
    \addlegendentry{BA Alg.}
    \addplot [color=mycolor6, line width=2.5pt, mark size=3.0pt, mark=otimes, mark options={solid, mycolor6}]
          table[row sep=crcr]{%
        0.0000    0.0000\\
    0.9801    0.1489\\
    1.5087    0.1821\\
    1.9230    0.1979\\
    2.3104    0.2070\\
        };
        \addlegendentry{Agg-IB Alg.}
    \addplot [color=mycolor2, line width=2.5pt, mark size=3.0pt, mark=diamond*, mark options={solid, mycolor2}]
          table[row sep=crcr]{%
        0.0000    0.0000\\
    0.9992    0.1523\\
    1.5835    0.1873\\
    1.9900    0.1996\\
    2.3207    0.2069\\
        };
        \addlegendentry{Seq-IB Alg.}
    \addplot [color=mycolor5, line width=2.5pt, mark size=3.0pt, mark=square*, mark options={solid, rotate=180, mycolor5}]
      table[row sep=crcr]{%
    0.0000    0.0000\\
    0.9928    0.1512\\
    0.9992    0.1523\\
    1.9537    0.2002\\
    2.3427    0.2037\\
    };
     \addlegendentry{Det-IB Alg.}
    \end{axis}
    \end{tikzpicture}%
        \caption{$\beta=0.6$.}
        \label{fig:beta_06}
        \end{subfigure}
        \quad
        \begin{subfigure}[b]{0.3\textwidth}
    \begin{tikzpicture}[scale=0.5]
    \pgfplotsset{every major grid/.append style={densely dashed}}
\begin{axis}[%
width=3.528in,
height=3.029in,
at={(1.011in,0.767in)},
scale only axis,
xmin=0,
xmax=3,
xlabel style={font=\color{white!15!black}},
xlabel={$I(X; T)$ (bits)},
ymin=0,
ymax=0.5,
ylabel style={font=\color{white!15!black}},
ylabel={$I(Y; T)$ (bits)},
axis background/.style={fill=white},
xmajorgrids,
ymajorgrids,
legend style={at={(0.64,0.0)}, anchor=south west, legend cell align=left, align=left, draw=white!15!black}
]
\addplot [color=mycolor1, line width=2.5pt, mark size=3.0pt, mark=+, mark options={solid, mycolor1}]
  table[row sep=crcr]{%
    0.0000    0.0000\\
    0.2115    0.1113\\
    0.6565    0.3010\\
    0.7863    0.3392\\
    0.8457    0.3526\\
    1.2974    0.4196\\
    1.5531    0.4407\\
    1.7396    0.4515\\
    1.8876    0.4582\\
    2.0089    0.4626\\
    2.1128    0.4658\\
    2.2038    0.4682\\
    2.2844    0.4701\\
    2.3572    0.4717\\
    2.4232    0.4729\\
    2.4837    0.4740\\
    2.5396    0.4749\\
    2.5915    0.4757\\
    2.6400    0.4763\\
    2.6855    0.4769\\
    2.7282    0.4774\\
    2.7687    0.4779\\
    2.8070    0.4783\\
    2.8434    0.4787\\
};
\addlegendentry{BA Alg.}
\addplot [color=mycolor6, line width=2.5pt, mark size=3.0pt, mark=otimes, mark options={solid, mycolor6}]
          table[row sep=crcr]{%
        0.0000         0\\
    0.9996    0.3685\\
    1.5027    0.4159\\
    1.9745    0.4550\\
    2.1886    0.4613\\
    2.4816    0.4670\\
    2.7747    0.4728\\
    2.9667    0.4772\\
        };
    \addlegendentry{Agg-IB Alg.}
    \addplot [color=mycolor2, line width=2.5pt, mark size=3.0pt, mark=diamond*, mark options={solid, mycolor2}]
          table[row sep=crcr]{%
         0.0000         0\\
    0.9889    0.3519\\
    1.5696    0.4325\\
    1.9630    0.4531\\
    2.2766    0.4639\\
    2.4741    0.4660\\
    2.8914    0.4766\\
        };
        \addlegendentry{Seq-IB Alg.}
  \addplot [color=mycolor5, line width=2.5pt, mark size=3.0pt, mark=square*, mark options={solid, rotate=180, mycolor5}]
      table[row sep=crcr]{%
    0.0000   0.0000\\
    0.9996    0.3685\\
    1.9425    0.4540\\
    1.9857    0.4550\\
    2.3245    0.4627\\
    2.7445    0.4579\\
    3.0086    0.4734\\
    };
     \addlegendentry{Det-IB Alg.}
\end{axis}
\end{tikzpicture}%
        \caption{$\beta=1$.}
        \label{fig:beta_1}
        \end{subfigure}
        \quad
        \begin{subfigure}[b]{0.3\textwidth}
        \begin{tikzpicture}[scale=0.5]
        \pgfplotsset{every major grid/.append style={densely dashed}}
        \begin{axis}[%
        width=3.528in,
        height=3.029in,
        at={(1.011in,0.767in)},
        scale only axis,
        xmin=0,
        xmax=2.5,
        xlabel style={font=\color{white!15!black}},
        xlabel={$I(X;  T)$ (bits)},
        ymin=0,
        ymax=0.71,
        ylabel style={font=\color{white!15!black}},
        ylabel={$I(Y; T)$ (bits)},
        axis background/.style={fill=white},
        xmajorgrids,
        ymajorgrids,
        legend style={at={(0.64,0.0)}, anchor=south west, legend cell align=left, align=left, draw=white!15!black}
        ]
        \addplot [color=mycolor1, line width=2.5pt, mark size=3.0pt, mark=+, mark options={solid, mycolor1}]
          table[row sep=crcr]{%
        0.0000    0.0000\\
            0.4778    0.3465\\
            0.7902    0.5329\\
            0.8445    0.5587\\
            0.8777    0.5727\\
            0.8995    0.5810\\
            0.9145    0.5862\\
            0.9253    0.5897\\
            0.9333    0.5921\\
            0.9396    0.5938\\
            0.9549    0.5977\\
            1.0000    0.6087\\
            1.0380    0.6174\\
            1.0702    0.6243\\
            1.0980    0.6300\\
            1.3868    0.6714\\
            1.5669    0.6861\\
            1.7019    0.6939\\
            1.8107    0.6988\\
            1.9032    0.7022\\
            1.9829    0.7047\\
            2.0533    0.7065\\
            2.1167    0.7080\\
            2.1737    0.7092\\
        };
        \addlegendentry{BA Alg.}
        \addplot [color=mycolor6, line width=2.5pt, mark size=3.0pt, mark=otimes, mark options={solid, mycolor6}]
          table[row sep=crcr]{%
        0.0000    0.0000\\
    0.9994    0.6005\\
    1.4597    0.6556\\
    1.8679    0.6918\\
    2.0261    0.7008\\
    2.1566    0.7054\\
    2.5000    0.7091\\
        };
    \addlegendentry{Agg-IB Alg.}
    \addplot [color=mycolor2, line width=2.5pt, mark size=3.0pt, mark=diamond*, mark options={solid, mycolor2}]
          table[row sep=crcr]{%
         0.0000    0.0000\\
    0.9997    0.6025\\
    1.4493    0.6496\\
    1.8249    0.6621\\
    2.1577    0.6794\\
    2.3958    0.7067\\
        };
        \addlegendentry{Seq-IB Alg.}
        \addplot [color=mycolor5, line width=2.5pt, mark size=3.0pt, mark=square*, mark options={solid, rotate=180, mycolor5}]
      table[row sep=crcr]{%
    0.0000    0.0000\\
    1.0000    0.6026\\
    1.6134    0.6832\\
    1.7578    0.6924\\
    2.0748    0.7019\\
    2.2689    0.7063\\
    2.5000   0.7069\\
    };
     \addlegendentry{Det-IB Alg.}
        \end{axis}
        \end{tikzpicture}%
        \caption{$\beta=\sqrt{2}$.}
        \label{fig:beta_sqrt2}
        \end{subfigure}
 \caption{The three baseline algorithms compared with the BA algorithm in terms of the objective mutual information $I(Y; T)$ and the constraint $I(X; T)$ for Bernoulli source and univariate mixture Gaussian observation when $\beta=\{0.6, 1, \sqrt{2}\}$.}
     \label{fig:beta_one_dim}
\end{figure}
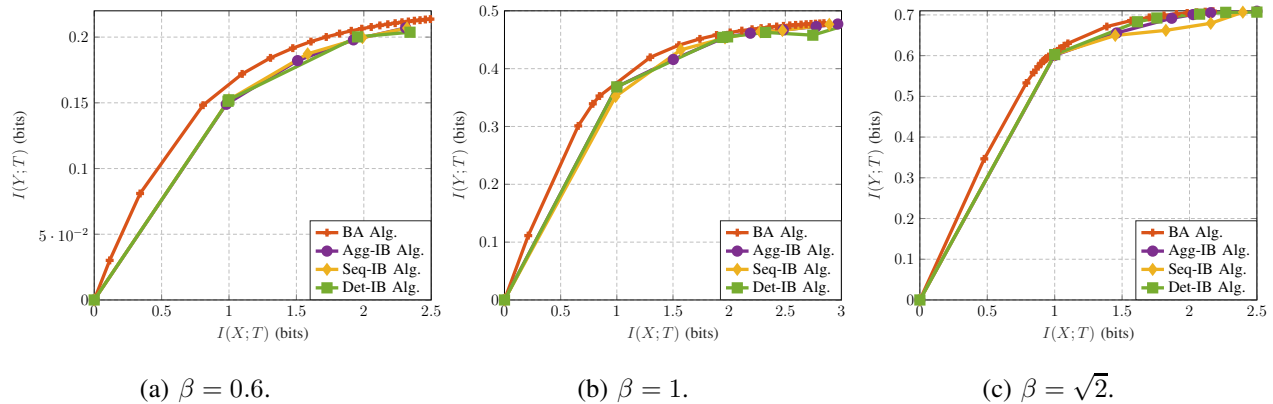

In this section, we perform numerical experiments to validate the optimal bound of the IB problem using the BA algorithm. We compare it to three baseline algorithms: Agg-IB (section \ref{sec: agg_IB}), Seq-IB (section \ref{sec:Seq-IB}), and Det-IB (section \ref{sec:det_IB}), considering the Bernoulli source labels and \emph{univariate} Gaussian mixture observations with different values of $\beta \in \{0.6, 1, \sqrt{2}\}$.
As illustrated in Fig. \ref{fig:beta_one_dim}, it shows that the performance of the three baseline algorithms is comparable, while the BA algorithm exhibits superior performance. This observation underscores the validity of the numerical optimal bound obtained with the BA algorithm.

\subsection{Simulations on the univariate mixture Gaussian IB problem}
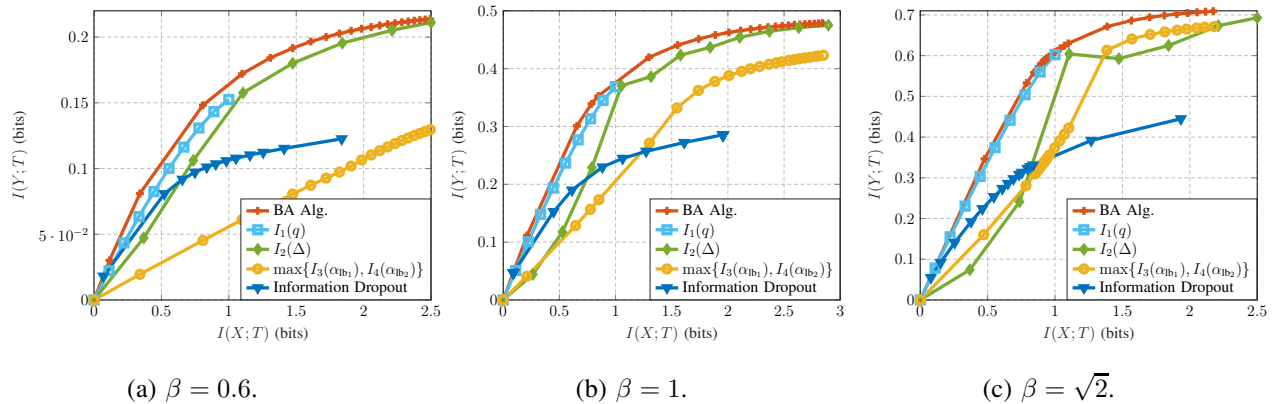
\begin{figure}[htb]
\centering
        \begin{subfigure}[b]{0.3\textwidth}
            \begin{tikzpicture}[scale=0.5]
            \pgfplotsset{every major grid/.append style={densely dashed}}
        \begin{axis}[%
        width=3.528in,
        height=3.029in,
        at={(1.011in,0.767in)},
        scale only axis,
        xmin=0,
        xmax=2.5,
        xlabel style={font=\color{white!15!black}},
        xlabel={$I(X; T)$ (bits)},
        ymin=0,
        ymax=0.22,
        ylabel style={font=\color{white!15!black}},
        ylabel={$I(Y; T)$ (bits)},
        axis background/.style={fill=white},
        xmajorgrids,
        ymajorgrids,
        legend style={at={(0.435, 0.00)}, anchor=south west, legend cell align=left, align=left, draw=white!15!black}
        ]
        \addplot [color=mycolor1, line width=2.5pt, mark size=3.0pt, mark=+, mark options={solid, mycolor1}]
          table[row sep=crcr]{%
           0.0000    0.0000\\
            0.1149    0.0301\\
            0.3405    0.0810\\
            0.8077    0.1482\\
            1.0976    0.1720\\
            1.3077    0.1842\\
            1.4732    0.1916\\
            1.6078    0.1966\\
            1.7214    0.2001\\
            1.8202    0.2027\\
            1.9070    0.2048\\
            1.9848    0.2064\\
            2.0550    0.2077\\
            2.1191    0.2089\\
            2.1780    0.2098\\
            2.2325    0.2106\\
            2.2832    0.2113\\
            2.3306    0.2119\\
            2.3751    0.2125\\
            2.4170    0.2129\\
            2.4566    0.2134\\
            2.4942    0.2137\\
        };
        \addlegendentry{BA Alg.}
        \addplot [color=mycolor4, line width=2.5pt, mark size=3.0pt, mark=square, mark options={solid, rotate=180, mycolor4}]
          table[row sep=crcr]{%
        0    0.0000\\
            0.1111    0.0222\\
            0.2222    0.0434\\
            0.3333    0.0635\\
            0.4444    0.0824\\
            0.5556    0.1001\\
            0.6667    0.1163\\
            0.7778    0.1308\\
            0.8889    0.1433\\
            1.0000    0.1525\\
        };
        \addlegendentry{$I_1(q)$}
        \addplot [color=mycolor5, line width=2.5pt, mark size=3.0pt, mark=diamond, mark options={solid, rotate=180, mycolor5}]
          table[row sep=crcr]{%
         0         0\\
            0.3684    0.0472\\
            0.7368    0.1062\\
            1.1053    0.1576\\
            1.4737    0.1801\\
            1.8421    0.1955\\
            2.2105    0.2052\\
            2.500    0.2111\\
        };
        \addlegendentry{$I_2(\Delta)$}
        
        \addplot [color=mycolor2, line width=2.5pt, mark size=3.0pt, mark=o, mark options={solid, rotate=180, mycolor2}]
          table[row sep=crcr]{%
        0         0
            0.1149    0.0067\\
            0.3405    0.0195\\
            0.8077    0.0453\\
            1.0976    0.0610\\
            1.3077    0.0721\\
            1.4732    0.0806\\
            1.6078    0.0873\\
            1.7214    0.0927\\
            1.8202    0.0973\\
            1.9070    0.1019\\
            1.9848    0.1065\\
            2.0550    0.1103\\
            2.1191    0.1136\\
            2.1780    0.1165\\
            2.2325    0.1190\\
            2.2832    0.1213\\
            2.3306    0.1233\\
            2.3751    0.1251\\
            2.4170    0.1267\\
            2.4566    0.1282\\
            2.4942    0.1296\\
        };
        \addlegendentry{$\max\{I_3(\alpha_{\text{lb}_1}) , I_4(\alpha_{\text{lb}_2}) \}$}
        \addplot [color=mycolor3, line width=2.5pt, mark size=3.0pt, mark=triangle, mark options={solid, rotate=180, mycolor3}]
          table[row sep=crcr]{%
        0.0689    0.0181\\
            0.5224    0.0808\\
            0.6556    0.0918\\
            0.7487    0.0972\\
            0.8362    0.1011\\
            0.9036    0.1035\\
            0.9817    0.1059\\
            1.0532    0.1078\\
            1.1566    0.1103\\
            1.2547    0.1123\\
            1.4102    0.1153\\
            1.8370    0.1226\\
        };
        \addlegendentry{Information Dropout}
        \end{axis}
        \end{tikzpicture}%
        \caption{$\beta=0.6$.}
        \label{fig:all_beta_06}
        \end{subfigure}
        ~~~~
        \begin{subfigure}[b]{0.3\textwidth}
            \begin{tikzpicture}[scale=0.5]
    \pgfplotsset{every major grid/.append style={densely dashed}}
\begin{axis}[%
        width=3.528in,
        height=3.029in,
        at={(1.011in,0.767in)},
        scale only axis,
        xmin=0,
        xmax=3,
        xlabel style={font=\color{white!15!black}},
        xlabel={$I(X; T)$ (bits)},
        ymin=0,
        ymax=0.5,
        ylabel style={font=\color{white!15!black}},
        ylabel={$I(Y; T)$ (bits)},
        axis background/.style={fill=white},
        xmajorgrids,
        ymajorgrids,
        legend style={at={(0.435, 0.00)}, anchor=south west, legend cell align=left, align=left, draw=white!15!black}
]
\addplot [color=mycolor1, line width=2.5pt, mark size=3.0pt, mark=+, mark options={solid, mycolor1}]
  table[row sep=crcr]{%
    0.0000    0.0000\\
    0.2115    0.1113\\
    0.6565    0.3010\\
    0.7863    0.3392\\
    0.8457    0.3526\\
    1.2974    0.4196\\
    1.5531    0.4407\\
    1.7396    0.4515\\
    1.8876    0.4582\\
    2.0089    0.4626\\
    2.1128    0.4658\\
    2.2038    0.4682\\
    2.2844    0.4701\\
    2.3572    0.4717\\
    2.4232    0.4729\\
    2.4837    0.4740\\
    2.5396    0.4749\\
    2.5915    0.4757\\
    2.6400    0.4763\\
    2.6855    0.4769\\
    2.7282    0.4774\\
    2.7687    0.4779\\
    2.8070    0.4783\\
    2.8434    0.4787\\
};
\addlegendentry{BA Alg.}
\addplot [color=mycolor4, line width=2.5pt, mark size=3.0pt, mark=square, mark options={solid, rotate=180, mycolor4}]
  table[row sep=crcr]{%
0    0.0000\\
    0.1111    0.0510\\
    0.2222    0.1005\\
    0.3333    0.1481\\
    0.4444    0.1936\\
    0.5556    0.2367\\
    0.6667    0.2768\\
    0.7778    0.3133\\
    0.8889    0.3452\\
    1.0000    0.3689\\
};
\addlegendentry{$I_1(q)$}
\addplot [color=mycolor5, line width=2.5pt, mark size=3.0pt, mark=diamond, mark options={solid, rotate=180, mycolor5}]
  table[row sep=crcr]{%
 0         0\\
    0.2632    0.0442\\
    0.5263    0.1176\\
    0.7895    0.2290\\
    1.0526    0.3706\\
    1.3158    0.3862\\
    1.5789    0.4234\\
    1.8421    0.4368\\
    2.1053    0.4542\\
    2.3684    0.4646\\
    2.6316    0.4710\\
    2.8947    0.4752\\
};
\addlegendentry{$I_2(\Delta)$}
\addplot [color=mycolor2, line width=2.5pt, mark size=3.0pt, mark=o, mark options={solid, rotate=180, mycolor2}]
  table[row sep=crcr]{%
 0    0.0000\\
    0.2103    0.0410\\
    0.6442    0.1286\\
    0.7760    0.1567\\
    0.8519    0.1734\\
    1.3017    0.2712\\
    1.5473    0.3322\\
    1.7371    0.3622\\
    1.8842    0.3779\\
    2.0094    0.3880\\
    2.1223    0.3954\\
    2.2011    0.3998\\
    2.2825    0.4038\\
    2.3665    0.4076\\
    2.4312    0.4102\\
    2.4758    0.4118\\
    2.5416    0.4142\\
    2.5883    0.4157\\
    2.6347    0.4171\\
    2.6809    0.4185\\
    2.7300    0.4199\\
    2.7769    0.4212\\
    2.8032    0.4218\\
    2.8510    0.4231\\
};
\addlegendentry{$\max\{I_3(\alpha_{\text{lb}_1}) , I_4(\alpha_{\text{lb}_2}) \}$}
\addplot [color=mycolor3, line width=2.5pt, mark size=3.0pt, mark=triangle, mark options={solid, rotate=180, mycolor3}]
  table[row sep=crcr]{%
0.0906    0.0468\\
    0.0906    0.0468\\
    0.0906    0.0468\\
    0.0906    0.0468\\
    0.4483    0.1527\\
    0.6124    0.1900\\
    0.8832    0.2299\\
    1.0633    0.2448\\
    1.2724    0.2570\\
    1.6110    0.2722\\
    1.9567    0.2856\\
    1.9567    0.2856\\
    1.9567    0.2856\\
    1.9567    0.2856\\
    1.9567    0.2856\\
    1.9567    0.2856\\
    1.9567    0.2856\\
    1.9567    0.2856\\
    1.9567    0.2856\\
    1.9567    0.2856\\
    1.9567    0.2856\\
    1.9567    0.2856\\
    1.9567    0.2856\\
    1.9567    0.2856\\
    1.9567    0.2856\\
};
\addlegendentry{Information Dropout}
\end{axis}
\end{tikzpicture}%
        \caption{$\beta=1$.}
        \label{fig:all_beta_1}
        \end{subfigure}
        ~~
        \begin{subfigure}[b]{0.3\textwidth}
        \begin{tikzpicture}[scale=0.5]
        \pgfplotsset{every major grid/.append style={densely dashed}}
        \begin{axis}[%
         width=3.528in,
        height=3.029in,
        at={(1.011in,0.767in)},
        scale only axis,
        xmin=0,
        xmax=2.5,
        xlabel style={font=\color{white!15!black}},
        xlabel={$I(X; T)$ (bits)},
        ymin=0,
        ymax=0.71,
        ylabel style={font=\color{white!15!black}},
        ylabel={$I(Y; T)$ (bits)},
        axis background/.style={fill=white},
        xmajorgrids,
        ymajorgrids,
        legend style={at={(0.435,0.0)}, anchor=south west, legend cell align=left, align=left, draw=white!15!black}
        ]
        \addplot [color=mycolor1, line width=2.5pt, mark size=3.0pt, mark=+, mark options={solid, mycolor1}]
          table[row sep=crcr]{%
        0.0000    0.0000\\
            0.4778    0.3465\\
            0.7902    0.5329\\
            0.8445    0.5587\\
            0.8777    0.5727\\
            0.8995    0.5810\\
            0.9145    0.5862\\
            0.9253    0.5897\\
            0.9333    0.5921\\
            0.9396    0.5938\\
            0.9549    0.5977\\
            1.0000    0.6087\\
            1.0380    0.6174\\
            1.0702    0.6243\\
            1.0980    0.6300\\
            1.3868    0.6714\\
            1.5669    0.6861\\
            1.7019    0.6939\\
            1.8107    0.6988\\
            1.9032    0.7022\\
            1.9829    0.7047\\
            2.0533    0.7065\\
            2.1167    0.7080\\
            2.1737    0.7092\\
        };
        \addlegendentry{BA Alg.}
        \addplot [color=mycolor4, line width=2.5pt, mark size=3.0pt, mark=square, mark options={solid, rotate=180, mycolor4}]
          table[row sep=crcr]{%
         0    0.0000\\
            0.1111    0.0783\\
            0.2222    0.1552\\
            0.3333    0.2304\\
            0.4444    0.3035\\
            0.5556    0.3741\\
            0.6667    0.4413\\
            0.7778    0.5039\\
            0.8889    0.5598\\
            1.0000    0.6026\\
        };
        \addlegendentry{$I_1(q)$}
        \addplot [color=mycolor5, line width=2.5pt, mark size=3.0pt, mark=diamond, mark options={solid, rotate=180, mycolor5}]
          table[row sep=crcr]{%
         0         0\\
            0.3684    0.0742\\
            0.7368    0.2411\\
            1.1053    0.6041\\
            1.4737    0.5926\\
            1.8421    0.6247\\
            2.2105    0.6734\\
            2.5    0.6929\\
        };
        \addlegendentry{$I_2(\Delta)$}
         \addplot [color=mycolor2, line width=2.5pt, mark size=3.0pt, mark=o, mark options={solid, rotate=180, mycolor2}]
          table[row sep=crcr]{%
        0.0003    0.0000  \\
            0.4709    0.1603\\
            0.7844    0.2805\\
            0.8537    0.3097\\
            0.8758    0.3192\\
            0.8973    0.3286\\
            0.9182    0.3379\\
            0.9182    0.3379\\
            0.9387    0.3470\\
            0.9387    0.3470\\
            0.9588    0.3561\\
            0.9973    0.3736\\
            1.0342    0.3906\\
            1.0692    0.4069\\
            1.1028    0.4224\\
            1.3844    0.6132\\
            1.5706    0.6411\\
            1.7060    0.6521\\
            1.8141    0.6583\\
            1.9056    0.6625\\
            1.9766    0.6653\\
            2.0504    0.6678\\
            2.1098    0.6698\\
            2.1805    0.6716\\
        };
        \addlegendentry{$\max\{I_3(\alpha_{\text{lb}_1}) , I_4(\alpha_{\text{lb}_2}) \}$}
        \addplot [color=mycolor3, line width=2.5pt, mark size=3.0pt, mark=triangle, mark options={solid, rotate=180, mycolor3}]
          table[row sep=crcr]{%
        0.0791    0.0549\\
            0.1518    0.0916\\
            0.2558    0.1409\\
            0.3771    0.1920\\
            0.4614    0.2241\\
            0.5463    0.2536\\
            0.6095    0.2737\\
            0.6508    0.2861\\
            0.6921    0.2976\\
            0.7336    0.3083\\
            0.7544    0.3133\\
            0.7958    0.3228\\
            0.8162    0.3271\\
            0.8363    0.3313\\
            1.2685    0.3917\\
            1.9336    0.4449\\
        };
        \addlegendentry{Information Dropout}
        \end{axis}
        \end{tikzpicture}%
        \caption{$\beta=\sqrt{2}$.}
        \label{fig:all_beta_sqrt2}
        \end{subfigure}
 \caption{The three achievable schemes compared with the BA algorithm and the information dropout method in terms of the objective mutual information $I(Y; T)$ and the constraint $I(X;T)$ for Bernoulli source and univariate mixture Gaussian observation when $\beta\in \{0.6, 1, \sqrt{2}\}$.}
     \label{fig:all_beta_one_dim}
\end{figure}


Next, we provide a comprehensive analysis of the performance of the three proposed achievable schemes as the signal-to-noise ratio (SNR) parameter $\beta$ varies. We present the comparisons in Figure~\ref{fig:all_beta_one_dim}, where we evaluate the three schemes proposed in Proposition~\ref{prop:one-bit}--\ref{theo:main-results} against the BA algorithm and the information dropout method for different values of $\beta \in \{0.6, 1, \sqrt{2}\}$. 
 For a fair comparison, in the information dropout method we use single-layer neural networks for both $f_1$ and $f_2$, i.e., $f_1(X) = \sigma(w_1 X) + 1$ and $f_2(X) = \sigma(w_2 X)$, where $\sigma(t) = (1+ \exp(-t))^{-1}$ denotes the logistic sigmoid function. A bias term $b = 1$ is introduced into $f_1(x)$ to avoid problems when calculating the conditional probability $f_{T|X}(t|x)$. The parameters can be optimized either by gradient descent or by brute search over $w_1$ and $w_2$ spaces based on problem \eqref{eq:obj_func_2}.

The simulations provide compelling insights, revealing that the combination of the three proposed schemes closely approximates the performance of the BA algorithm and yields better results compared to the information dropout method \cite{achille2018information}.
The information dropout method shows comparable performance to the proposed approach in the small $R$ region, but deteriorates for larger $R$.
This also shows that within the information dropout framework, the single hidden layer NN model, despite being universal approximators with a sufficiently large number of neurons \cite{hornik1991approximation}, is less efficient in solving the IB problem. 
This is also (empirically) supported by the fact that a certain (large) value of mutual information $I(X; T)$ \emph{cannot} be achieved with the single-layer information dropout approach in Figure~\ref{fig:beta_one_dim}.


For the scheme using two-level quantization in Proposition~\ref{prop:one-bit}, recall that the observation denoted as $X = \beta Y + \epsilon$ in \eqref{eq:def_model_scalar}, a larger SNR $\beta$ leads to a larger separation between the means of the mixture Gaussian distribution. 
In this case, the two-level quantization (indicator function) already provides a good estimate of $Y$. As $\beta$ increases, the simulations show that $I_1(q)$ approaches the performance of the BA algorithm. Additionally, in the region with a smaller constraint on $I(X; T)$, it is observed that $I_1(q)$ outperforms other methods such as $I_2(\Delta)$, indicating that two-level quantization combined with a random variable following a Bernoulli distribution performs better than other methods, such as deterministic quantization with a quantization level $L = 2$, when $I(X; T) \leq 1$ bit. This observation is due to the fact that for a small value of $I(X; T)$, it is more effective to directly estimate the source $Y$ directly, and the two-level quantization function can provide a reliable estimate in such cases.

In contrast, for the scheme using deterministic quantization, $I_2(\Delta)$ converges to the BA algorithm as $I(X; T)$ increases. 
Furthermore, the gap between the BA algorithm and $I_2(\Delta)$ is relatively small compared to $\max\{I_3(\alpha_{\text{lb}_1}), I_4(\alpha_{\text{lb}_2})\}$ when $\beta \in \{0.6, 1\}$. However, when $\beta$ is large (e.g., $\beta = \sqrt{2}$), $I_2(\Delta)$ performs similarly to $\max\{I_3(\alpha_{\text{lb}_1}), I_4(\alpha_{\text{lb}_2})\}$.

It is also worth noting that the scheme using ``soft'' quantization is sensitive to the value of $\beta$ because it is derived through variational optimization, where a Bernoulli distribution is introduced as the variational distribution. As $\beta$ increases, the introduced distribution becomes closer to the variational distribution, reducing the gap between them. Therefore, when $\beta$ is small (e.g., $\beta = 0.6$), the penalty incurred by introducing the variational distribution is already significant, resulting in a lower rate $I(Y; T)$. Conversely, as $\beta$ increases, $\max\{I_3(\alpha_{\text{lb}_1}), I_4(\alpha_{\text{lb}_2})\}$ approaches the performance of the BA algorithm, even performing better than $I_2(\Delta)$ when $\beta=\sqrt{2}$ for large $R$.

\subsection{Simulation on multivariate mixture Gaussian IB problem}
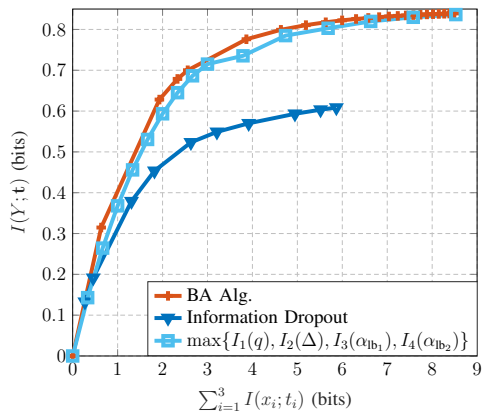
\begin{figure}[ht]
 \centering
\begin{tikzpicture}[scale=0.6]
\pgfplotsset{every major grid/.append style={densely dashed}}
\begin{axis}[%
width=3.528in,
height=3.029in,
at={(1.011in,0.767in)},
scale only axis,
xmin=0,
xmax=9,
xlabel style={font=\color{white!15!black}},
xlabel={$\sum_{i=1}^3 I(x_i; t_i)$ (bits)},
ymin=0,
ymax=0.85,
ylabel style={font=\color{white!15!black}},
ylabel={$I(Y; \mathbf{t})$ (bits)},
axis background/.style={fill=white},
xmajorgrids,
ymajorgrids,
legend style={at={(0.185,0.0)}, anchor=south west, legend cell align=left, align=left, draw=white!15!black}
]
\addplot [color=mycolor1, line width=2.5pt, mark size=3.0pt, mark=+, mark options={solid, mycolor1}]
  table[row sep=crcr]{%
0.0000    0.0000    \\
    0.6346    0.3147\\
    1.9297    0.6282\\
    2.3322    0.6788\\
    2.5606    0.7000\\
    3.8656    0.7760\\
    4.6335    0.7991\\
    5.1905    0.8106\\
    5.6321    0.8176\\
    5.9960    0.8223\\
    6.3063    0.8257\\
    6.5787    0.8282\\
    6.8200    0.8302\\
    7.0374    0.8318\\
    7.2348    0.8332\\
    7.4160    0.8343\\
    7.5830    0.8352\\
    7.7383    0.8360\\
    7.8834    0.8368\\
    8.0192    0.8374\\
    8.1473    0.8379\\
    8.2682    0.8384\\
    8.3828    0.8389\\
    8.4916    0.8393\\
};
\addlegendentry{BA Alg.}

\addplot [color=mycolor3, line width=2.5pt, mark size=3.0pt, mark=triangle, mark options={solid, rotate=180, mycolor3}]
  table[row sep=crcr]{%
0.2721    0.1327    \\
    0.4573    0.1914\\
    1.3083    0.3791\\
    1.8205    0.4535\\
    2.6312    0.5233\\
    3.2098    0.5490\\
    3.9163    0.5700\\
    4.9462    0.5933\\
    5.5165    0.6034\\
    5.8676    0.6085\\
};
\addlegendentry{Information Dropout}
\addplot [color=mycolor4, line width=2.5pt, mark size=3.0pt, mark=square, mark options={solid, rotate=180, mycolor4}]
  table[row sep=crcr]{%
 0    0.0000        \\
    0.3333    0.1429\\
    0.6667    0.2643\\
    1.0000    0.3678\\
    1.3333    0.4559\\
    1.6667    0.5305\\
    2.0000    0.5933\\
    2.3333    0.6451\\
    2.6667    0.6864\\
    3.0000    0.7151\\
    3.7895    0.7354\\
    4.7368    0.7849\\
    5.6842    0.8029\\
    6.6316    0.8197\\
    7.5789    0.8298\\
    8.5263    0.8364\\
};
\addlegendentry{$\max\{I_1(q), I_2(\Delta), I_3(\alpha_{\text{lb}_1}) , I_4(\alpha_{\text{lb}_2}) \}$}
\end{axis}
\end{tikzpicture}%
    \caption{The three methods compared with the respect to the objective mutual information $I(Y; \mathbf{t})$ and the constraint $\sum_{i=1}^3 I(x_i; t_i)$ for Bernoulli source and three-dimensional mixture multivariate Gaussian observation when $\boldsymbol{\beta}=[0.9, 1, 1.1]^{\rm T}$.}
    \label{fig:rand_p_3}
\end{figure}

Next, Figure~\ref{fig:rand_p_3} extends the above experiments to multivariate setting with $\mathbf{\beta} = [0.9, 1.0, 1.1]^{\rm T}$, by solving the IB problem in an entry-wise manner. Moreover, for the rate allocation in~\eqref{eq:rate allocation}, we set $R_1 = R_2 = R_3 = \frac{R}{3}$ in this section when we consider the case of $d_0=3$. From simulation, we consistently observe a close match between our proposed unified lower bound in Theorem \ref{theo:main-results-2} and the numerically optimal BA solution for all $R$ range. It indicates the good performance of our proposed methods in the multivariate mixture Gaussian IB problem.

\subsection{Application to Gaussian mixture classification with information leakage}

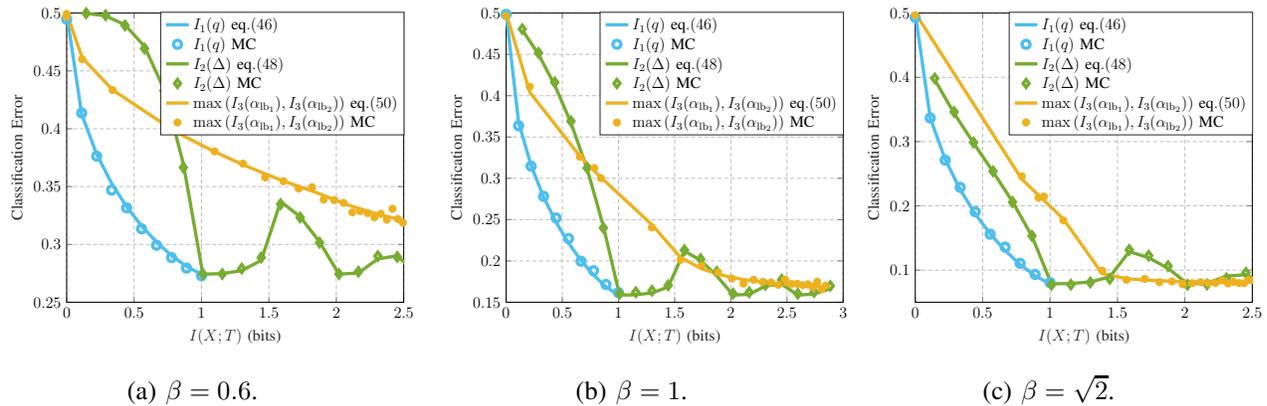
\begin{figure}[htb]
\centering
        \begin{subfigure}[b]{0.3\textwidth}
            \begin{tikzpicture}[scale=0.5]
            \pgfplotsset{every major grid/.append style={densely dashed}}
        \begin{axis}[%
        width=3.528in,
        height=3.029in,
        at={(1.011in,0.767in)},
        scale only axis,
        xmin=0,
        xmax=2.5,
        xlabel style={font=\color{white!15!black}},
        xlabel={$I(X; T)$ (bits)},
        ymin=0.25,
        ymax=0.5,
        ylabel style={font=\color{white!15!black}},
        ylabel={Classification Error},
        axis background/.style={fill=white},
        xmajorgrids,
        ymajorgrids,
        legend style={at={(0.278, 0.578)}, anchor=south west, legend cell align=left, align=left, draw=white!15!black}
        ]
        \addplot [color=mycolor4, line width=2.5pt]
  table[row sep=crcr]{%
0	0.499779544060303\\
0.111111111111111	0.412558218293543\\
0.222222222222222	0.378036195986585\\
0.333333333333333	0.352791407423213\\
0.444444444444444	0.332637773217223\\
0.555555555555555	0.315973112603746\\
0.666666666666667	0.302015711210671\\
0.777777777777778	0.290372020739398\\
0.888888888888889	0.280932243798087\\
1	0.274253131205636\\
};
\addlegendentry{$I_1(q) ~\text{eq}.(46)$}

\addplot [color=mycolor4, line width=2.5pt, only marks, mark size=3.0pt, mark=o, mark options={solid, mycolor4}]
  table[row sep=crcr]{%
0	0.49475\\
0.111111111111111	0.4136\\
0.222222222222222	0.37625\\
0.333333333333333	0.34695\\
0.444444444444444	0.3316\\
0.555555555555555	0.31345\\
0.666666666666667	0.2993\\
0.777777777777778	0.28855\\
0.888888888888889	0.2795\\
1	0.2729\\
};
\addlegendentry{$I_1(q) ~\text{MC}$}

\addplot [color=mycolor5, line width=2.5pt]
  table[row sep=crcr]{%
0.144269504088896	0.499783084466616\\
0.288539008177793	0.497600115329128\\
0.432808512266689	0.489852078734856\\
0.577078016355585	0.47118283892512\\
0.721347520444482	0.434307691728253\\
0.865617024533378	0.369758211923734\\
1.00988652862227	0.274253672678103\\
1.15415603271117	0.274725283029626\\
1.29842553680007	0.277383463583397\\
1.44269504088896	0.286133158622636\\
1.58696454497786	0.336951680458498\\
1.73123404906676	0.324072822302355\\
1.87550355315565	0.302171924350469\\
2.01977305724455	0.274256821961319\\
2.16404256133344	0.275286235698208\\
2.30831206542234	0.287513324934679\\
2.45258156951124	0.290144787410922\\
2.59685107360013	0.274254398877182\\
2.74112057768903	0.275775963047195\\
2.88539008177793	0.283562842808178\\
};
\addlegendentry{$I_2(\Delta) ~\text{eq}.(48)$}

\addplot [color=mycolor5, line width=2.5pt, only marks, mark size=3.0pt, mark=diamond, mark options={solid, mycolor5}]
  table[row sep=crcr]{%
0.144269504088896	0.49985\\
0.288539008177793	0.49815\\
0.432808512266689	0.4893\\
0.577078016355585	0.46935\\
0.721347520444482	0.4329\\
0.865617024533378	0.36655\\
1.00988652862227	0.2745\\
1.15415603271117	0.2744\\
1.29842553680007	0.27915\\
1.44269504088896	0.28865\\
1.58696454497786	0.3348\\
1.73123404906676	0.32325\\
1.87550355315565	0.3015\\
2.01977305724455	0.2744\\
2.16404256133344	0.2764\\
2.30831206542234	0.2899\\
2.45258156951124	0.28915\\
2.59685107360013	0.27435\\
2.74112057768903	0.27735\\
2.88539008177793	0.2828\\
};
\addlegendentry{$I_2(\Delta) ~\text{MC}$}

\addplot [color=mycolor2, line width=2.5pt]
  table[row sep=crcr]{%
0	0.500000000012091\\
0	0.500000000012091\\
0	0.500000000012091\\
0	0.500000000012091\\
0.114002998680255	0.461567249477031\\
0.339961014602799	0.43362186403067\\
0.80729498398942	0.397387872705982\\
1.09716500977113	0.380080302638198\\
1.30741509067472	0.368967148380725\\
1.47283089622864	0.360959312994021\\
1.60749734345604	0.354900442480609\\
1.72114059216659	0.350105932401724\\
1.81985598854554	0.346176598684244\\
1.90673336403522	0.342408045697381\\
1.9845050748607	0.338640003072072\\
2.05456070535312	0.335549485371382\\
2.11875057117985	0.332940442006109\\
2.17769143370411	0.330712739316117\\
2.23212022981408	0.328785117379143\\
2.28288069995413	0.327090119177679\\
2.33022273097674	0.325591815411626\\
2.37474875528803	0.324250274394368\\
2.41666209004643	0.323043521796699\\
2.45625276035645	0.321950618077179\\
2.49389424731899	0.32095150694215\\
};
\addlegendentry{$\max{(I_3(\alpha_{\rm lb_1}), I_3(\alpha_{\rm lb_2}))} ~\text{eq}.(50)$}

\addplot [color=mycolor2, line width=2.5pt, only marks, mark size=3.0pt, mark=asterisk, mark options={solid, mycolor2}]
  table[row sep=crcr]{%
0	0.5036\\
0	0.49675\\
0	0.5066\\
0	0.49925\\
0.114002998680255	0.4602\\
0.339961014602799	0.4334\\
0.80729498398942	0.3997\\
1.09716500977113	0.38045\\
1.30741509067472	0.36985\\
1.47283089622864	0.35785\\
1.60749734345604	0.35465\\
1.72114059216659	0.34825\\
1.81985598854554	0.34945\\
1.90673336403522	0.3387\\
1.9845050748607	0.3383\\
2.05456070535312	0.3359\\
2.11875057117985	0.3277\\
2.17769143370411	0.3287\\
2.23212022981408	0.32705\\
2.28288069995413	0.3234\\
2.33022273097674	0.3264\\
2.37474875528803	0.32135\\
2.41666209004643	0.3308\\
2.45625276035645	0.322\\
2.49389424731899	0.3187\\
};
\addlegendentry{$\max{(I_3(\alpha_{\rm lb_1}), I_3(\alpha_{\rm lb_2}))} ~\text{MC}$}
        \end{axis}
        \end{tikzpicture}%
        \caption{$\beta=0.6$.}
        \label{fig:scalar_beta_06}
        \end{subfigure}
        ~~~~
        \begin{subfigure}[b]{0.3\textwidth}
            \begin{tikzpicture}[scale=0.5]
            \pgfplotsset{every major grid/.append style={densely dashed}}
        \begin{axis}[%
        width=3.528in,
        height=3.029in,
        at={(1.011in,0.767in)},
        scale only axis,
        xmin=0,
        xmax=3,
        xlabel style={font=\color{white!15!black}},
        xlabel={$I(X; T)$ (bits)},
        ymin=0.15,
        ymax=0.5,
        ylabel style={font=\color{white!15!black}},
        ylabel={Classification Error},
        axis background/.style={fill=white},
        xmajorgrids,
        ymajorgrids,
        legend style={at={(0.278, 0.578)}, anchor=south west, legend cell align=left, align=left, draw=white!15!black}
        ]
        \addplot [color=mycolor4, line width=2.5pt]
  table[row sep=crcr]{%
0	0.499666655521417\\
0.111111111111111	0.367782037674714\\
0.222222222222222	0.315582375731676\\
0.333333333333333	0.277410526553364\\
0.444444444444444	0.246936851427318\\
0.555555555555555	0.221738743312903\\
0.666666666666667	0.200634196544362\\
0.777777777777778	0.183028139142351\\
0.888888888888889	0.168754549931013\\
1	0.158655274277189\\
};
\addlegendentry{$I_1(q) ~\text{eq}.(46)$}

\addplot [color=mycolor4, line width=2.5pt, only marks, mark size=3.0pt, mark=o, mark options={solid, mycolor4}]
  table[row sep=crcr]{%
0	0.4984\\
0.111111111111111	0.36365\\
0.222222222222222	0.3147\\
0.333333333333333	0.2781\\
0.444444444444444	0.252\\
0.555555555555555	0.2271\\
0.666666666666667	0.1996\\
0.777777777777778	0.18815\\
0.888888888888889	0.17155\\
1	0.1616\\
};
\addlegendentry{$I_1(q) ~\text{MC}$}

\addplot [color=mycolor5, line width=2.5pt]
  table[row sep=crcr]{%
0.144269504088896	0.479586453222068\\
0.288539008177793	0.450016131569332\\
0.432808512266689	0.412787139789554\\
0.577078016355585	0.367210180280151\\
0.721347520444482	0.311477914790254\\
0.865617024533378	0.242063590385247\\
1.00988652862227	0.158655712471965\\
1.15415603271117	0.159045448053978\\
1.29842553680007	0.161244231921461\\
1.44269504088896	0.168509335123827\\
1.58696454497786	0.212002698920247\\
1.73123404906676	0.2007080695099\\
1.87550355315565	0.181961169680267\\
2.01977305724455	0.158658314639107\\
2.16404256133344	0.159509156888061\\
2.30831206542234	0.1696595790653\\
2.45258156951124	0.171856168674092\\
2.59685107360013	0.158656312500931\\
2.74112057768903	0.159914119871718\\
2.88539008177793	0.16637044281925\\
};
\addlegendentry{$I_2(\Delta) ~\text{eq}.(48)$}

\addplot [color=mycolor5, line width=2.5pt, only marks, mark size=3.0pt, mark=diamond, mark options={solid, mycolor5}]
  table[row sep=crcr]{%
0.144269504088896	0.4806\\
0.288539008177793	0.45185\\
0.432808512266689	0.41625\\
0.577078016355585	0.3692\\
0.721347520444482	0.3125\\
0.865617024533378	0.24\\
1.00988652862227	0.1606\\
1.15415603271117	0.16215\\
1.29842553680007	0.164\\
1.44269504088896	0.17045\\
1.58696454497786	0.21305\\
1.73123404906676	0.20225\\
1.87550355315565	0.18685\\
2.01977305724455	0.16045\\
2.16404256133344	0.16205\\
2.30831206542234	0.17105\\
2.45258156951124	0.17735\\
2.59685107360013	0.1603\\
2.74112057768903	0.16235\\
2.88539008177793	0.1698\\
};
\addlegendentry{$I_2(\Delta) ~\text{MC}$}

\addplot [color=mycolor2, line width=2.5pt]
  table[row sep=crcr]{%
0	0.500000000001134\\
0	0.500000000001134\\
0.210559046480231	0.404161295315946\\
0.65612107749372	0.326960756999615\\
0.786094724699554	0.308857526539917\\
0.844271277156787	0.301042590284405\\
1.29713465701329	0.243972975274596\\
1.55318914424261	0.204127375609763\\
1.73941970659625	0.190576757961458\\
1.88710956469081	0.183922838531014\\
2.00872658841071	0.180037288144034\\
2.11243503740855	0.177488025187126\\
2.20372823352108	0.175669027309145\\
2.28445190388856	0.174316627142804\\
2.3569464839616	0.173266680123756\\
2.42316729763905	0.172420403467228\\
2.48358083711705	0.171728083100806\\
2.53957090841506	0.171145204736799\\
2.59147631321695	0.170649197212492\\
2.63982129265944	0.170221362315364\\
2.68544551085108	0.169844774052706\\
2.72820800954494	0.169513547647083\\
2.76852191008644	0.169218911609061\\
2.80685084172017	0.168953395440041\\
2.84337904470547	0.168712636890749\\
};
\addlegendentry{$\max{(I_3(\alpha_{\rm lb_1}), I_3(\alpha_{\rm lb_2}))} ~\text{eq}.(50)$}

\addplot [color=mycolor2, line width=2.5pt, only marks, mark size=3.0pt, mark=asterisk, mark options={solid, mycolor2}]
  table[row sep=crcr]{%
0	0.5036\\
0	0.49675\\
0.210559046480231	0.4109\\
0.65612107749372	0.32605\\
0.786094724699554	0.31235\\
0.844271277156787	0.3002\\
1.29713465701329	0.2405\\
1.55318914424261	0.20135\\
1.73941970659625	0.1943\\
1.88710956469081	0.1862\\
2.00872658841071	0.17875\\
2.11243503740855	0.1725\\
2.20372823352108	0.17725\\
2.28445190388856	0.1747\\
2.3569464839616	0.17395\\
2.42316729763905	0.1711\\
2.48358083711705	0.1775\\
2.53957090841506	0.1723\\
2.59147631321695	0.1723\\
2.63982129265944	0.1707\\
2.68544551085108	0.17165\\
2.72820800954494	0.1685\\
2.76852191008644	0.17485\\
2.80685084172017	0.16605\\
2.84337904470547	0.17005\\
};
\addlegendentry{$\max{(I_3(\alpha_{\rm lb_1}), I_3(\alpha_{\rm lb_2}))} ~\text{MC}$}
        \end{axis}
        \end{tikzpicture}%
        \caption{$\beta=1$.}
        \label{fig:scalar_beta_01}
        \end{subfigure}
        ~~
        \begin{subfigure}[b]{0.3\textwidth}
            \begin{tikzpicture}[scale=0.5]
            \pgfplotsset{every major grid/.append style={densely dashed}}
        \begin{axis}[%
        width=3.528in,
        height=3.029in,
        at={(1.011in,0.767in)},
        scale only axis,
        xmin=0,
        xmax=2.5,
        xlabel style={font=\color{white!15!black}},
        xlabel={$I(X; T)$ (bits)},
        ymin=0.05,
        ymax=0.5,
        ylabel style={font=\color{white!15!black}},
        ylabel={Classification Error},
        axis background/.style={fill=white},
        xmajorgrids,
        ymajorgrids,
        legend style={at={(0.278, 0.578)}, anchor=south west, legend cell align=left, align=left, draw=white!15!black}
        ]
        \addplot [color=mycolor4, line width=2.5pt]
  table[row sep=crcr]{%
0	0.499588525003443\\
0.111111111111111	0.336792300193569\\
0.222222222222222	0.272357887451397\\
0.333333333333333	0.225239178372949\\
0.444444444444444	0.187622969996251\\
0.555555555555555	0.156518835637337\\
0.666666666666667	0.130467726748843\\
0.777777777777778	0.108735096168946\\
0.888888888888889	0.0911160100614699\\
1	0.0786496286395833\\
};
\addlegendentry{$I_1(q) ~\text{eq}.(46)$}

\addplot [color=mycolor4, line width=2.5pt, only marks, mark size=3.0pt, mark=o, mark options={solid, mycolor4}]
  table[row sep=crcr]{%
0	0.4932\\
0.111111111111111	0.33675\\
0.222222222222222	0.2711\\
0.333333333333333	0.22875\\
0.444444444444444	0.1908\\
0.555555555555555	0.1559\\
0.666666666666667	0.1358\\
0.777777777777778	0.1105\\
0.888888888888889	0.09325\\
1	0.08045\\
};
\addlegendentry{$I_1(q) ~\text{MC}$}

\addplot [color=mycolor5, line width=2.5pt]
  table[row sep=crcr]{%
0.144269504088896	0.396372885715774\\
0.288539008177793	0.343841455722058\\
0.432808512266689	0.297802254530086\\
0.577078016355585	0.253134353937415\\
0.721347520444482	0.206461459361139\\
0.865617024533378	0.15265868899\\
1.00988652862227	0.0786501380808615\\
1.15415603271117	0.0791032575998983\\
1.29842553680007	0.0816153316081668\\
1.44269504088896	0.0894611582446619\\
1.58696454497786	0.128752102482781\\
1.73123404906676	0.119322013920949\\
1.87550355315565	0.102668130727416\\
2.01977305724455	0.0786531718153224\\
2.16404256133344	0.0796392041698296\\
2.30831206542234	0.0906493983878761\\
2.45258156951124	0.0928836926125854\\
2.59685107360013	0.0786508376457654\\
2.74112057768903	0.0801044927923543\\
2.88539008177793	0.0872154688281702\\
};
\addlegendentry{$I_2(\Delta) ~\text{eq}.(48)$}

\addplot [color=mycolor5, line width=2.5pt, only marks, mark size=3.0pt, mark=diamond, mark options={solid, mycolor5}]
  table[row sep=crcr]{%
0.144269504088896	0.3984\\
0.288539008177793	0.3462\\
0.432808512266689	0.29855\\
0.577078016355585	0.2538\\
0.721347520444482	0.2059\\
0.865617024533378	0.15325\\
1.00988652862227	0.0772\\
1.15415603271117	0.0773\\
1.29842553680007	0.0808\\
1.44269504088896	0.08645\\
1.58696454497786	0.1315\\
1.73123404906676	0.12255\\
1.87550355315565	0.1063\\
2.01977305724455	0.0772\\
2.16404256133344	0.07785\\
2.30831206542234	0.0872\\
2.45258156951124	0.0959\\
2.59685107360013	0.07735\\
2.74112057768903	0.0785\\
2.88539008177793	0.0884\\
};
\addlegendentry{$I_2(\Delta) ~\text{MC}$}

\addplot [color=mycolor2, line width=2.5pt]
  table[row sep=crcr]{%
0	0.500000000000091\\
0	0.500000000000091\\
0.790010679488529	0.24154309755638\\
0.91437860552295	0.216633807186716\\
0.954494425507215	0.208588471696238\\
1.09792987216179	0.179806237302644\\
1.38690105833974	0.0944042102082987\\
1.56697942742025	0.0867300175239209\\
1.70212795993419	0.0844449491054197\\
1.81108057797773	0.083383557467571\\
1.90341020310563	0.0827610712040976\\
1.98335588044765	0.082348088893307\\
2.05344969929154	0.0820524947063605\\
2.1169489766851	0.0818248069169309\\
2.17398747365761	0.0816459543277416\\
2.22657712717342	0.0814987240879987\\
2.275078881155	0.0813756484427494\\
2.32008550443042	0.0812708736898983\\
2.36217190243042	0.0811801266537437\\
2.40151963951326	0.0811009235386332\\
2.43895397300169	0.08103013890706\\
2.47419996013059	0.0809672016011671\\
2.50755916828383	0.0809106776217417\\
2.53932486005861	0.0808593992598296\\
2.56991725055214	0.0808121988977206\\
};
\addlegendentry{$\max{(I_3(\alpha_{\rm lb_1}), I_3(\alpha_{\rm lb_2}))} ~\text{eq}.(50)$}

\addplot [color=mycolor2, line width=2.5pt, only marks, mark size=3.0pt, mark=asterisk, mark options={solid, mycolor2}]
  table[row sep=crcr]{%
0	0.5036\\
0	0.49675\\
0.790010679488529	0.24615\\
0.91437860552295	0.2125\\
0.954494425507215	0.21455\\
1.09792987216179	0.17735\\
1.38690105833974	0.09925\\
1.56697942742025	0.08485\\
1.70212795993419	0.0865\\
1.81108057797773	0.08125\\
1.90341020310563	0.0824\\
1.98335588044765	0.07785\\
2.05344969929154	0.07965\\
2.1169489766851	0.08135\\
2.17398747365761	0.07975\\
2.22657712717342	0.07985\\
2.275078881155	0.08395\\
2.32008550443042	0.07915\\
2.36217190243042	0.08185\\
2.40151963951326	0.07985\\
2.43895397300169	0.0796\\
2.47419996013059	0.08495\\
2.50755916828383	0.08265\\
2.53932486005861	0.0809\\
2.56991725055214	0.0819\\
};
\addlegendentry{$\max{(I_3(\alpha_{\rm lb_1}), I_3(\alpha_{\rm lb_2}))} ~\text{MC}$}
        \end{axis}
        \end{tikzpicture}%
        \caption{$\beta=\sqrt{2}$.}
        \label{fig:scalar_beta_sqrt2}
        \end{subfigure}
 \caption{The three achievable schemes compared in terms of the classification error ${\rm{Pr}}(\widehat{Y} \neq Y)$ and the information leakage $I(X;T)$ for Bernoulli source and univariate mixture Gaussian observation when $\beta\in \{0.6, 1, \sqrt{2}\}$. Monte Carlo (MC) simulations are obtained by averaging over  independent runs. }
     \label{fig:scalar_mixture_gaussian_sim}
\end{figure}

\subsubsection{Simulation on the univariate observations}
In this section, we perform simulations for the binary classification problem with scalar observations to validate Propositions~\ref{prop:one-bit_2}~--~\ref{prop:soft_quantization}. 
As shown in Fig.~\ref{fig:scalar_mixture_gaussian_sim}, the solid line represents the closed-form misclassification error rates given in Propositions~\ref{prop:one-bit_2}~--~\ref{prop:soft_quantization}, while the marker points denote results obtained by Monte Carlo (MC) simulations over independent runs. 
It can be seen that the marker points perfectly match the corresponding solid line, confirming the results in Propositions~\ref{prop:one-bit_2}--\ref{prop:soft_quantization} on the precise tradeoff between misclassification error rates and information leakage.

\subsubsection{Simulation on the multivariate observations}
\begin{figure}[htb]
\centering
        \begin{subfigure}[b]{0.3\textwidth}
            \begin{tikzpicture}[scale=0.5]
            \pgfplotsset{every major grid/.append style={densely dashed}}
        \begin{axis}[%
        width=3.528in,
        height=3.029in,
        at={(1.011in,0.767in)},
        scale only axis,
        xmin=0,
        xmax=5,
        xlabel style={font=\color{white!15!black}},
        xlabel={$I(\x; \boldsymbol{t})$ (bits)},
        ymin=0,
        ymax=0.5,
        ylabel style={font=\color{white!15!black}},
        ylabel={Classification Error},
        axis background/.style={fill=white},
        xmajorgrids,
        ymajorgrids,
        legend style={at={(0.435, 0.715)}, anchor=south west, legend cell align=left, align=left, draw=white!15!black}
        ]
        \addplot [color=mycolor4, line width=2.5pt, mark size=3.0pt, mark=square, mark options={solid, rotate=180, mycolor4}]
          table[row sep=crcr]{%
        0.0056    0.4972\\
    0.2993    0.3286\\
    0.5491    0.2548\\
    0.7591    0.2156\\
    0.9211    0.1827\\
    1.1512    0.1566\\
    1.3757    0.1343\\
    1.4918    0.1180\\
    1.6686    0.1033\\
    1.8184    0.0937\\
        };
        \addlegendentry{$I_1(q)$}
        \addplot [color=mycolor5, line width=2.5pt, mark size=3.0pt, mark=diamond, mark options={solid, rotate=180, mycolor5}]
          table[row sep=crcr]{%
         0.0001    0.5000\\
    0.4784    0.3631\\
    1.0093    0.1998\\
    1.6961    0.1008\\
    2.1874    0.0851\\
    2.7694    0.0746\\
    3.2647    0.0667\\
    3.8428    0.0627\\
    4.3977    0.0616\\
    4.9622    0.0569\\
        };
        \addlegendentry{$I_2(\Delta)$}
        
        \addplot [color=mycolor2, line width=2.5pt, mark size=3.0pt, mark=o, mark options={solid, rotate=180, mycolor2}]
          table[row sep=crcr]{%
        0.0135    0.4913\\
    0.5896    0.2147 \\
       1.1201    0.1584\\
     1.5728    0.1198\\
    2.0641    0.1001\\
    2.7000    0.0861\\
    3.3854    0.0794\\
    3.5052    0.0707\\
    3.5389    0.0693\\
    4.0812    0.0665\\
        };
        \addlegendentry{$\max\{I_3(\alpha_{\text{lb}_1}) , I_4(\alpha_{\text{lb}_2}) \}$}
        \addplot [color=mycolor3, line width=2.5pt, mark size=3.0pt, mark=triangle, mark options={solid, rotate=180, mycolor3}]
          table[row sep=crcr]{%
        0.6477    0.3448\\
    0.6525    0.3400\\
    0.9825    0.3259\\
    1.8509    0.2921\\
    2.5625    0.2758\\
    3.1745    0.2743\\
    4.0150    0.2717\\
        };
        \addlegendentry{Information Dropout}
        \end{axis}
        \end{tikzpicture}%
        \caption{Class $7$ versus class $9$.}
        \label{fig: MNIST_79}
        \end{subfigure}
        ~~~~
        \begin{subfigure}[b]{0.3\textwidth}
            \begin{tikzpicture}[scale=0.5]
    \pgfplotsset{every major grid/.append style={densely dashed}}
\begin{axis}[%
        width=3.528in,
        height=3.029in,
        at={(1.011in,0.767in)},
        scale only axis,
        xmin=0,
        xmax=4,
        xlabel style={font=\color{white!15!black}},
        xlabel={$I(\x; \boldsymbol{t})$ (bits)},
        ymin=0,
        ymax=0.5,
        ylabel style={font=\color{white!15!black}},
        ylabel={Classification Error},
        axis background/.style={fill=white},
        xmajorgrids,
        ymajorgrids,
        legend style={at={(0.435, 0.715)}, anchor=south west, legend cell align=left, align=left, draw=white!15!black}
]
\addplot [color=mycolor4, line width=2.5pt, mark size=3.0pt, mark=square, mark options={solid, rotate=180, mycolor4}]
  table[row sep=crcr]{%
0.0008    0.4863\\
    0.2458    0.3037\\
    0.5228    0.2327\\
    0.6690    0.1871\\
    0.8625    0.1593\\
    1.0779    0.1265\\
    1.2794    0.1099\\
    1.3667    0.0847\\
    1.5620    0.0639\\
    1.6994    0.0509\\
};
\addlegendentry{$I_1(q)$}
\addplot [color=mycolor5, line width=2.5pt, mark size=3.0pt, mark=diamond, mark options={solid, rotate=180, mycolor5}]
  table[row sep=crcr]{%
 0.0001    0.5000\\
    0.6773    0.2426\\
    1.1851    0.1263\\
    1.6733    0.0579\\
    2.1504    0.0374\\
    2.7057    0.0315\\
    3.1895    0.0282\\
    4.0   0.0256\\
};
\addlegendentry{$I_2(\Delta)$}
\addplot [color=mycolor2, line width=2.5pt, mark size=3.0pt, mark=o, mark options={solid, rotate=180, mycolor2}]
  table[row sep=crcr]{%
0.0001    0.5000\\
    0.5357    0.1891\\
    1.1979    0.0540\\
    1.4644    0.0524\\
    1.9280    0.0468\\
    2.59522   0.03818\\
    3.1390    0.0365\\
    3.3371    0.0341\\
    4.0000    0.0320\\
};
\addlegendentry{$\max\{I_3(\alpha_{\text{lb}_1}) , I_4(\alpha_{\text{lb}_2}) \}$}
\addplot [color=mycolor3, line width=2.5pt, mark size=3.0pt, mark=triangle, mark options={solid, rotate=180, mycolor3}]
  table[row sep=crcr]{%
0.6406    0.3095\\
    0.6276    0.3188\\
    0.9146    0.2928\\
    1.3502    0.2795\\
    1.8213    0.2691\\
    2.3661    0.2436\\
    3.0547    0.2392\\
    3.2003    0.2342\\
    3.7574    0.2298\\
    3.8086    0.2358\\
    4.0000    0.2346\\
};
\addlegendentry{Information Dropout}
\end{axis}
\end{tikzpicture}%
        \caption{Class $4$ versus class $1$.}
        \label{fig:41_MNIST}
        \end{subfigure}
        ~~
        \begin{subfigure}[b]{0.3\textwidth}
        \begin{tikzpicture}[scale=0.5]
        \pgfplotsset{every major grid/.append style={densely dashed}}
        \begin{axis}[%
         width=3.528in,
        height=3.029in,
        at={(1.011in,0.767in)},
        scale only axis,
        xmin=0,
        xmax=3,
        xlabel style={font=\color{white!15!black}},
        xlabel={$I(\x; \boldsymbol{t})$ (bits)},
        ymin=0,
        ymax=0.5,
        ylabel style={font=\color{white!15!black}},
        ylabel={Classification Error},
        axis background/.style={fill=white},
        xmajorgrids,
        ymajorgrids,
        legend style={at={(0.435, 0.715)}, anchor=south west, legend cell align=left, align=left, draw=white!15!black}
        ]
        \addplot [color=mycolor4, line width=2.5pt, mark size=3.0pt, mark=square, mark options={solid, rotate=180, mycolor4}]
          table[row sep=crcr]{%
         0.0002    0.4990\\
    0.2815    0.3501\\
    0.5033    0.2844\\
    0.8156    0.2401\\
    1.0349    0.2059\\
    1.2639    0.1750\\
    1.4043    0.1487\\
    1.6412    0.1290\\
    1.8299    0.1122\\
    1.9948    0.1010\\
        };
        \addlegendentry{$I_1(q)$}
        \addplot [color=mycolor5, line width=2.5pt, mark size=3.0pt, mark=diamond, mark options={solid, rotate=180, mycolor5}]
          table[row sep=crcr]{%
         0.0001    0.5000\\
    0.3773    0.3573\\
    0.5596    0.2499\\
    1.0838    0.1272\\
    2.1508    0.1042\\
    2.9313    0.1041\\
    3.0    0.1007\\
        };
        \addlegendentry{$I_2(\Delta)$}
         \addplot [color=mycolor2, line width=2.5pt, mark size=3.0pt, mark=o, mark options={solid, rotate=180, mycolor2}]
          table[row sep=crcr]{%
       0.0001    0.5000\\
       0.0037    0.4965\\
       0.9214    0.1543\\
    0.9851    0.1304\\
    1.2964    0.1122\\
    1.3722    0.1062\\
    1.6481    0.1050\\
    2.1786    0.1015\\
    2.2183    0.1010\\
    2.6957    0.1007\\
        };
        \addlegendentry{$\max\{I_3(\alpha_{\text{lb}_1}) , I_4(\alpha_{\text{lb}_2}) \}$}
        \addplot [color=mycolor3, line width=2.5pt, mark size=3.0pt, mark=triangle, mark options={solid, rotate=180, mycolor3}]
          table[row sep=crcr]{%
        0.4950    0.3725\\
    0.7132    0.3718\\
    1.0773    0.3279\\
    1.5009    0.3132\\
    2.0516    0.3099\\
    2.6279    0.3077\\
    2.7072    0.3016\\
    2.7391    0.3002\\
        };
        \addlegendentry{Information Dropout}
        \end{axis}
        \end{tikzpicture}%
        \caption{Even class versus odd class.}
        \label{fig:even_odd_MNIST}
        \end{subfigure}
 \caption{The three achievable schemes compared to the information dropout method with the respect to the classification error $\rm{Pr}(\widehat{Y} \neq Y)$ and the constraint $I(\x; \boldsymbol{t})$ for dimension-reduced MNIST data.}
     \label{fig:beta_MNIST_all}
\end{figure}
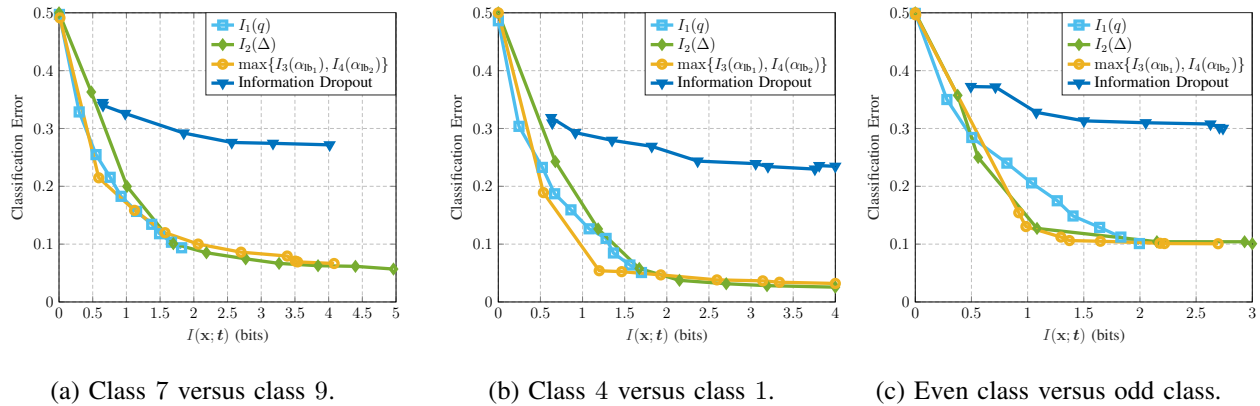

We further provide simulations on multivariate IB problem given in~\eqref{eq:obj_func_entry}, with symmetric Bernoulli $Y$ and $\mathbf{x}|Y \sim \mathcal N(\bbeta Y, 1)$, where we construct the Markov chain $Y \rightarrow {\mathbf{x}} \rightarrow \mathbf{t}$. 
The logistic regression output of the intermediate representation $\mathbf{t}$ is used as the estimator for $Y$, defined as 
\begin{equation}
\widehat{Y}=
    \begin{cases}
        1 & \text{if } \sigma(\mathbf{w}^{\rm T} \mathbf{t} + {b}) \geq 0.5,\\
        -1 & \text{if } \sigma(\mathbf{w}^{\rm T} \mathbf{t} + {b}) < 0.5,
    \end{cases}
\end{equation}
where $\sigma(\cdot)$ is the sigmoid function, $\mathbf{w}$ is the weight vector, and ${b}$ is the bias term.
These logistic regression parameters $\mathbf{w}$ and $\mathbf{b}$ are determined by training on a training data. 
And the misclassification error rate can then be obtained from {an independent test set}.

Precisely, here we apply the analytic IB schemes derived from Propositions~\ref{prop:one-bit} through \ref{theo:main-results} (in fact, their multivariate versions as described in Section~\ref{sec: application}) to real-world data from the MNIST database~\cite{lecun1998gradient}. 
Due to the computational complexity in estimating mutual information of high-dimensional random vectors, we first reduce the dimensionality of the vectorized MNIST images by randomly projecting them through a Gaussian matrix (which is known to preserve the Euclidean distances between high-dimensional data vectors, see for example the popular Johnson--Lindenstrauss lemma~\cite{johnson1984extensions} and an overview of randomized sketching methods in~\cite{drineas2016randnla}).
This results in features of dimension $d_0 = 3$. 
{The implementations of other iterative algorithms require the estimation of the conditional probability $p_{\mathbf{x}|Y}(\mathbf{x}|y)$ from data samples and are therefore not included here for the sake of fair comparison.}
In this experiment, we only compare the three schemes proposed in propositions~\ref{prop:one-bit} to \ref{theo:main-results} with the information dropout approach.

Fig.~\ref{fig:beta_MNIST_all} illustrates the ``accuracy-complexity'' tradeoff between the proposed analytical IB scheme and the information dropout approach. 
This comparison is based on the classification error $\rm{Pr}(\widehat{Y} \neq Y)$ plotted against the information leakage budget $I(\mathbf{x}; \mathbf{t})$ on the reduced MNIST features\footnote{See Appendix \ref{sec:MNIST_processing} for details on MNIST data preprocessing.}. 
We use the jackknife approach \cite{Jackknife2011zeng} to numerically estimate the mutual information $I(\mathbf{x}; \mathbf{t})$ from the available MNIST image samples. For better visualization, linear interpolation is used to estimate the maximum of the two lower bounds proposed in Proposition~\ref{theo:main-results}. Notably, the proposed analytical IB scheme consistently shows advantageous performance on real (and non-Gaussian) data, 
suggesting a potentially broader applicability of the proposed approach.

\section{Conclusion}
\label{sec: conclusion}
In conclusion, this paper has contributed by deriving achievable solutions for the information bottleneck (IB) problem with Bernoulli sources and Gaussian mixture data, using both soft and deterministic quantization schemes. Using the Blahut-Arimoto algorithm, an approximately optimal solution is obtained, and the results have been extended to the vector-mixed Gaussian observation problem. Through extensive experiments conducted on the proposed achievable schemes under various signal-to-noise ratios (SNRs), our theoretical framework has been robustly validated. Looking ahead, an intriguing avenue for future research is to determine the distribution of the input $Y$ for the observation model $X = \beta Y + \epsilon$ as defined in \eqref{eq:def_model_scalar}, maximizing the IB while ensuring that $I(X; T) \leq R$ and subject to a unit variance constraint. In particular, it has been conjectured in previous work \cite{zaidientropy2020} that the optimal $Y$ is discrete. The insights gained from the present study, particularly with respect to the IB for Gaussian mixture observations, may serve as a valuable tool in delineating the precise low SNR range where the symmetric binary input under consideration proves to be optimal.


\appendices
\section{Proof of Proposition \ref{prop:one-bit}}
\label{sec:proof_one_bit}
According to the findings in \cite{zaidientropy2020}, the optimal design of the representation of $\overline{X} $ for DSBS $Y$ and $\overline{X} $ is \emph{explicitly} given by:
\begin{align}
T = \overline{X} \oplus \overline{N}, \text{where $\overline{N} \sim {\rm Bern}(q)$ for some $q \in [0,1]$},
\end{align}
which aligns with the form in \eqref{eq:t} and $\oplus$ denotes the exclusive `or' operation. Hence, the parameter $q$ can be obtained by setting $I(\overline{X}; T)$ as $R$ in \eqref{eq:q}. 
By the above construction, we denote the achieved 
$
I(Y;T) $ by $I_1(q)$, which is equal to 
$
 \ln 2 - H(p(1-q) + q(1-p)). 
$
Note that
$p$ represents the miss detection probability, i.e., $\mathbb{P}(\overline{X}  = 1|Y = -1)$, so $p = \int_{0}^{\infty} \frac{1}{\sqrt{2 \pi}} \exp(-\frac{(x + \beta)^2}{2}) dx$.
This concludes the proof of Proposition~\ref{prop:one-bit}.

\section{Proof of Proposition \ref{theo:main-results}}
\label{sec:proof_main}
Since it is challenging to directly solve \eqref{eq:V}, we derive a lower bound to \eqref{eq:V} by introducing an upper bound to $I(X; T)$ to obtain $\alpha$. Therefore, we first compute $I(X; T)$.
Since $\widehat{X}$ is a one-to-one mapping of $X$, it is evident that $I(X; T)$ and $I(\widehat{X}; T)$ are equal.
Since $T|\widehat{X}$ is a Gaussian distribution with unit variance, the conditional differential entropy is $h(T|\widehat{X}) = \frac{1}{2} \ln( 2\pi e)$.
In addition, we can compute the pdf of $\widehat{X}$ as 
\begin{align}\label{eq:p_yhat}
&p_{\widehat{X}} (\widehat{x}) \!=\! p_X(x) \frac{\partial x}{\partial \widehat{x}}, \nonumber \\
    \!&=\!  \frac{1}{ \beta \sqrt{2 \pi}} \exp{(\!-\frac{(1/\beta \tanh^{-1}(\widehat{x}))^2 \!+\!\beta^2}{2})} \frac{1}{(1 \!-\! \widehat{x}^2)^{1.5}}.
\end{align}

According to the information inequality~\cite{hanineq}, for any probability distribution $q_T(t)$, an upper bound to $h(T)$ is given by
\begin{align}
    h(T) = -\int p_T(t) \ln(p_T(t)) dt \leq -\int p_T(t) \ln(q_T(t)) dt.
    \label{eq:ht}
\end{align}
Then by \eqref{eq:ht}, an upper bound of $I(\widehat{X}; T)$ based on the variational distribution $q_T(t)$ is derived as \eqref{eq:upper bound of I hat based on q}.
Moreover, the distribution of $T$ is much complicated due to the distribution of $\widehat{X}$ in \eqref{eq:p_yhat}. Therefore, instead of introducing variational distribution of $T$, we come up with the variational distribution of $\widehat{X}$. Since $\widehat{X}$ is the MMSE estimation of $Y$ given observation $X$, for simplicity,  we design the variational distribution of $\widehat{X}$ as Bernoulli distribution, i.e., $q_{\widehat{X}}(\widehat{X}=-1) = q_{\widehat{X}}(\widehat{X}=1) = \frac{1}{2}$. Intuitively speaking, the less the noise power of $X$ is, the closer the variational distribution $q_{\widehat{X}}$ gets to the true distribution $p_{\widehat{X}}$.
Therefore, the variational distribution of $T$ is given by \eqref{eq:qt}.
Hence, by taking~\eqref{eq:qt} into~\eqref{eq:upper bound of I hat based on q}, an upper bound to $I(\widehat{X}; T)$ is given by 
\begin{subequations}
    \begin{align} 
    &I(\widehat{X}; T) \leq -\int_{-\infty}^{\infty} \left(\int_{-1}^{1} p_{T|\widehat{X}} (t |\widehat{x}) p_{\widehat{X}} (\widehat{x}) d \widehat{x}\right) \ln q_T(t) dt  -\frac{1}{2} \ln( 2\pi e)  \\
    &= \frac{{\alpha}^2 -1}{2} + {\int_{-1}^{1} p_{\widehat{X}} (\widehat{x}) \left(\int_{-\infty}^{\infty} p_{T|\widehat{X}}(t|\widehat{x}) \frac{t^2}{2} dt \right) d\widehat{x}} \nonumber \\
    &\quad - \int_{-\infty}^{\infty} \left(\int_{-1}^{1} p_{T|\widehat{X}} (t |\widehat{x}) p_{\widehat{X}} (\widehat{x}) d\widehat{x}\right) \ln(\cosh(\alpha t)) dt.
\label{eq:three_parts}
\end{align}
\end{subequations}
Since $T|\widehat{X}$ follows a Gaussian distribution $\mathcal{N}(\alpha \widehat{X}, 1)$, then ${\int_{-1}^{1} p_{\widehat{X}} (\widehat{x}) \left(\int_{-\infty}^{\infty} p_{T|\widehat{X}}(t|\widehat{x}) \frac{t^2}{2} dt \right) d\widehat{x}}$ in \eqref{eq:three_parts} is given by 
    \begin{align}
\int_{-1}^{1} p_{\widehat{X}} (\widehat{x}) &\left(\int_{-\infty}^{\infty} p_{T|\widehat{X}}(t|\widehat{x}) \frac{t^2}{2} dt \right) d\widehat{x}  = \int_{-1}^1 p_{\widehat{X}} (\widehat{x}) \frac{1 + {\alpha}^2 \widehat{x}^2}{2} d\widehat{x}. \label{eq:2_term}
\end{align}
By taking \eqref{eq:2_term} into \eqref{eq:three_parts}, and based on the notations of $f(\beta)$ and $g(\beta)$, an upper bound to $I(\widehat{X}; T)$ based on variational distribution is given by \eqref{eq:UB}.

Therefore, in the following we will propose {\bf two lower bounds} of $\ln(\cosh{\alpha t})$ 
  to derive a loosen upper bound to $I(\widehat{X}; T)$ with respect to~\eqref{eq:UB}.   

{\bf First lower bound of $\ln(\cosh{\alpha t})$:}
According to the inequality $\ln(\cosh(x)) \geq \sqrt{1 + x^2} -1, ~\forall~ x \geq 0$, and some important inequalities, i.e., the convexity of the function, and Jensen's inequality,  $(d)$ in \eqref{eq:UB} is upper bounded by \eqref{eq:first upper bound of Ixt}. $\alpha$ can be obtained by forcing the RHS of \eqref{eq:first upper bound of Ixt} to equal $R$, i.e., 
\begin{align}
     \frac{{\alpha}^2}{2} (1 + f(\beta)) -\sqrt{1 + {\alpha}^4  (g(\beta))^2} + 1=R.
    \label{eq:first upper bound equation}
\end{align}

The next step is to solve the equation~\eqref{eq:first upper bound equation}.
Assuming that $x = \alpha^2$
, $a = \frac{(1 +f(\beta))^2 - 4 (g(\beta))^2}{4}$, $b = (1 - R)(1 +f(\beta))$, $c = R^2 - 2R$. 
and $\Delta = b^2 - 4 ac$. First in order to check whether there exists a real solution, we need to check whether $\Delta$ is always non-negative when $R \geq 0$. Then we have 
\begin{subequations}
  \begin{align}
     \Delta &= (1 +f(\beta))^2 + 4 (g(\beta))^2 (R^2 - 2R)  \\
    &\geq (1 +f(\beta))^2 + 4 (g(\beta))^2 (-1)  \\
    &= (1 +f(\beta) - 2 (g(\beta)))(1 +f(\beta) + 2(g(\beta))). \label{eq:delta}
\end{align}  
\end{subequations}
Note that the term $1 +f(\beta) + 2(g(\beta))$ in \eqref{eq:delta} is always non-negative, and based on $\int_{-1}^0 2p_{\widehat{X}} (\widehat{x}) d \widehat{x} = 1$, the term $1 +f(\beta) - 2 (g(\beta))$ can be further developed as 
\begin{subequations}
    \begin{align}
   1 +f(\beta) \!-\! 2 (g(\beta)) \!&=\! 1 \!+\! 2\int_{-1}^0 p_{\widehat{X}} (\widehat{x}) \widehat{x}^2 d \widehat{x} \!+\! 4\int_{-1}^0 p_{\widehat{X}} (\widehat{x}) \widehat{x} d \widehat{x}, \\
   &=1 + 2\int_{-1}^0 p_{\widehat{X}} (\widehat{x})  \left[(\widehat{x}+1)^2 -1\right] d \widehat{x}, \\
   &> 1 + 2\int_{-1}^0 p_{\widehat{X}} (\widehat{x}) (-1) d \widehat{x},  \\
   &= 0. \label{eq:inequality}
\end{align}
\end{subequations}
Hence, the   term  $1 +f(\beta) - 2 (g(\beta))$ is always positive, and thus   $\Delta \geq 0$ holds when $R \geq 0$. Therefore, there always exists some real solution of \eqref{eq:first upper bound equation}.

Secondly, we need to check whether there exists a positive solution in problem \eqref{eq:first upper bound equation}. From \eqref{eq:inequality}, it can be seen that $a$ is always positive. 
 When $0 \leq R \leq 1$, $b$ is also positive. In this way, we need to compare $-b$ and $\sqrt{\Delta}$, so we have 
\begin{align}\label{eq:b_delta}
     b^2 - \Delta &= (R^2 - 2R)\underbrace{\left[\right(1 +f(\beta)))^2 - 4 (g(\beta))^2]}_{(e)}.
\end{align}
Therefore, since $(e)$ in \eqref{eq:b_delta}  is always positive, when $0\leq R \leq 2$, $R^2 - 2R$ is non-positive, it results in $|b| \leq \sqrt{\Delta}$ while $R \geq 2$,  it comes to $|b| \geq \sqrt{\Delta}$. 

As a result, when $R \leq 1$,  we have $|b| \leq \sqrt{\Delta}$ and  $-b +\sqrt{\Delta} \geq 0$; thus there   exists one positive and real solution of \eqref{eq:first upper bound equation}, which is 
\begin{align}
    \alpha_{\text{lb}_1} = \sqrt{\frac{-b + \sqrt{\Delta}}{2 a}}.
\end{align}
 In addition, when
  $R> 1$, we have  $-b$ is positive and $\sqrt{\Delta}$ is also positive; thus there always exists some positive solution of \eqref{eq:first upper bound equation}. However, it may exist two positive solutions. Since the larger correlation factor $\alpha$ will result in the larger $I(Y; T)$, we will choose the larger solution when two positive solutions occur. Therefore, the solution is also 
\begin{align}
    \alpha_{\text{lb}_1} = \sqrt{\frac{-b + \sqrt{\Delta}}{2 a}}. 
\end{align}

{\bf Second lower bound of $\ln(\cosh{\alpha t})$:}  Based on $\ln(\cosh(x)) \geq x - \ln2, ~\forall~ x \geq 0$, an upper bound to $(d)$ in \eqref{eq:UB} is given by \eqref{eq:LB_2_derivation_2}. According to the upper bound on $\mathbb{P} (S \geq - \alpha \hat{x})$ in \eqref{eq:x_exp_bound}, the bound is further relaxed as \eqref{eq:UB_2_constraint_UB}. 
 $\alpha$ can be obtained by forcing the RHS of \eqref{eq:UB_2_constraint_UB} to equal $R$, i.e., 
\begin{align}
     \alpha^2\left[\frac{1}{2} + \frac{f(\beta)}{2} -g(\beta)\right]  + \ln2\overset{\Delta}{=}R.
    \label{eq: Bound_2_equality}
\end{align}
According to \eqref{eq:inequality}, when $R \geq \ln2$, there exists a positive solution to \eqref{eq: Bound_2_equality}, which is 
\begin{align}\label{eq:alpha_2}
    \alpha_{\text{lb}_2} = \sqrt{\frac{R - \ln2}{\frac{1}{2} + \frac{f(\beta)}{2} - g(\beta)}}
\end{align}

In the end, through \eqref{eq:I_y_t} we can compute the lower bound on $I(Y; T)$, i.e., $I_3(\alpha_{\text{lb}_1})$, and $I_4(\alpha_{\text{lb}_2})$, respectively.
\section{Achievability Proof of~\eqref{eq:obj_func_entry}}
\label{sec:achievability proof}
In order to prove that a solution of~\eqref{eq:obj_func_entry} is also a solution of~\eqref{eq:obj_func_1}, we only need to prove 
\begin{align}
  I(\x; \bt) \leq   \sum_{i\in [1:d_0]} I(x_i;t_i), \label{eq:app to be proved}
\end{align}
where $t_i~|~x_i \sim \mathcal N \left(\alpha_i \tanh(\beta x_i), 1\right)$. 
This is because by~\eqref{eq:obj_func_entry}, we have $\sum_{i\in [1:d_0]} I(x_i;t_i)=\sum_{i\in[1:d_0]}R_i=R$. If~\eqref{eq:app to be proved} holds, we also have 
$I(\x; \bt)\leq R$, coinciding with the secrecy constraint in~\eqref{eq:security constr}.
In the rest of this section, we will prove~\eqref{eq:app to be proved}. 

By our construction in~\eqref{eq:obj_func_entry}, it can be seen that 
for each $i\in [1:d_0]$, we have the following Markov chain
\begin{align}
    (x_1,t_1,x_2,t_2,\ldots,x_{i-1},t_{i-1},x_{i+1},t_{i+1},\ldots,x_{d_0},t_{d_0}) \longrightarrow x_i \longrightarrow t_i. \label{eq:markov chain}
\end{align}
By the chain rule of mutual information,  we have 
\begin{align}
    I(\x; \bt)= I(\x;t_1)+ I(\x;t_2| t_1)+ \cdots+ I(\x;t_{d_0}|t_1,\ldots,t_{d_0-1}).\label{eq:chain rule of MI}
\end{align}
We then focus on each term on the RHS of~\eqref{eq:chain rule of MI}. For each $i\in [1:d_0]$, we have 
\begin{subequations}
    \begin{align}
    &I(\x;t_i| t_1,\ldots,t_{i-1}) \nonumber \\
    &= I(x_i;t_i|t_1,\ldots,t_{i-1})+I(x_1,\ldots,x_{i-1},x_{i+1},\ldots,x_{d_0};t_i|x_i,t_1,\ldots,t_{i-1})\\
    & =I(x_i;t_i|t_1,\ldots,t_{i-1}) \label{eq:by markov} \\
    &\leq I(x_i,t_1,\ldots,t_{i-1};t_i) \\
    &=I(x_i;t_i) +I(t_1,\ldots,t_{i-1};t_i|x_i) \\
    &=I(x_i;t_i), \label{eq:by markov again}
\end{align}
\end{subequations}
where~\eqref{eq:by markov} and~\eqref{eq:by markov again} come from the Markov chain~\eqref{eq:markov chain}.
By taking~\eqref{eq:by markov again} into~\eqref{eq:chain rule of MI}, we can directly prove~\eqref{eq:app to be proved}.

\section{proof of Proposition \ref{prop:one-bit_2}}
\label{sec:proof_one_bit_2}
Based on the formulated  Markov chain, $Y \rightarrow \overline{X} \rightarrow T$, where $\overline{X} = \mathbbm{1}_{X \geq 0}$ and $T = \overline{X} \oplus \overline{N}$, and given the estimator as 
\begin{equation}
\widehat{Y}=
    \begin{cases}
        1 & \text{if } T = 1,\\
        -1 & \text{if } T = 0,
    \end{cases}
\end{equation}
the classification error of the this scheme is defined as 
    \begin{align}
         {\rm Pr}(Y \neq \widehat{Y}) &=  \frac12  P_{T|Y}(t=1|y=-1) + \frac12  P_{T|Y}(t=0|y=1) \nonumber \\
         &=\frac12 \sum_{\overline{X} \in \{0, 1\}} P_{T, \overline{X}|Y}(t=1, \overline{X}|y=-1) + \frac12 \sum_{\overline{X} \in \{0, 1\}}  P_{T, \overline{X} |Y}(t=0, \overline{X}  |y=1) \\
         &= \frac12 \sum_{\overline{X} \in \{0, 1\}} P_{T|\overline{X}}(t=1|\overline{X}) P_{\overline{X}|Y} ( \overline{X}|y=-1)  \nonumber \\
         &\quad + \frac12 \sum_{\overline{X} \in \{0, 1\}} P_{T|\overline{X}}(t=0|\overline{X}) P_{\overline{X}|Y} ( \overline{X}|y=1)   \label{eq:error_binary_chain}\\
         &= (1- p)q  + p (1- q),
    \end{align}
    where \eqref{eq:error_binary_chain} holds due the Markov chain.

\section{Proof of Proposition~\ref{prop:derterm_Q_2}}
\label{sec:proof_determ_Q_2}
Based on the Markov chain $Y \rightarrow {X} \rightarrow T$, where $T = \widehat{Q}(X)$, and given the estimator as 
\begin{equation}
\widehat{Y}=
    \begin{cases}
        1 & \text{if } T \geq 0,\\
        -1 & \text{if } T < 0,
    \end{cases}
\end{equation}
the classification error of he multi-level deterministic quantization is defined as 
    \begin{align}
        {\rm Pr}(Y \neq \widehat{Y}) &= \frac{1}{2} {\rm Pr} (\widehat{Y}=1|Y =-1) + \frac{1}{2} {\rm Pr} (\widehat{Y}=-1|Y =1) \nonumber \\
        &= \frac{1}{2} {\rm Pr} ( T \geq 0|Y =-1) + \frac{1}{2} {\rm Pr} (T<0|Y =1) \nonumber \\
        &= \frac{1}{2}\sum_{T \geq 0} \mathbb{P}(T |Y =-1) + \frac{1}{2}\sum_{T < 0} \mathbb{P}(T |Y =1).
    \end{align}
Assuming that with the quantization points for $T$ $t_1 \leq t_2  \cdots \leq t_{L}$, the $s$ index indicates the subscript of the quantization point which itself is less than zero, while the next one of which is larger than zero, i.e., $t_s <0$ and $t_{s+1} \geq 0$, then the classification error is given by 
\begin{align}
    {\rm Pr}(Y \neq \widehat{Y}) = \frac{1}{2}\sum_{j=s+1}^L \mathbb{P}(T = t_j|Y =-1) +  \frac{1}{2}\sum_{j=1}^s \mathbb{P}(T = t_j|Y =1),
\end{align}
where the conditional probability $\mathbb{P}(T = t_j|Y)$ is defined in \eqref{eq:det_p_t_y}. Hence the classification error can be further derived as 
\begin{align}
     {\rm Pr}(Y \neq \widehat{Y}) &= \frac{1}{2} \left(Q(q_0 - \beta ) - Q(q_s - \beta)\right) + \frac{1}{2} \left(Q(q_s + \beta) - Q(q_L + \beta) \right) \nonumber \\
     &=\frac12 \left(1 - Q(q_s - \beta) + Q(q_s + \beta)\right) \nonumber \\
     &= \frac12 \left(Q(-q_s + \beta) + Q(q_s + \beta)\right) \label{eq:error_multi_a},
\end{align}
where \eqref{eq:error_multi_a} holds according to the property of $Q$ function, $Q(x) = 1- Q(-x)$.

\section{Proof of Proposition~\ref{prop:soft_quantization}}
\label{sec:proof_soft_quantization}
Based on the Markov chain, $Y \rightarrow {X} \rightarrow T$, where $T = \alpha \widehat{X} + \widehat{N} =\alpha \tanh{(\beta X)} + \widehat{N}$, and given the estimator as 
\begin{equation}
\widehat{Y}=
    \begin{cases}
        1 & \text{if } T \geq 0,\\
        -1 & \text{if } T < 0,
    \end{cases}
\end{equation}
the classification error of the this scheme is defined as 
    \begin{align}
         {\rm Pr}(Y \neq \widehat{Y}) &=  \frac{1}{2} {\rm Pr} (\widehat{Y}=1|Y =-1) + \frac{1}{2} {\rm Pr} (\widehat{Y}=-1|Y =1) \nonumber \\
         &= \frac{1}{2} {\rm Pr} (T \geq 0|Y =-1) + \frac{1}{2} {\rm Pr} (T <0|Y =1) \nonumber \\
         & = \frac{1}{2} \int_{0}^{\infty} p_{T|Y}(t|y =-1) + \frac{1}{2} \int_{-\infty}^{0} p_{T|Y}(t|y =1), \label{eq:error_a} \end{align}
    where $p_{T|Y}$ is defined in \eqref{eq:conditional proba}.
    Therefore, \eqref{eq:error_a} can be further derived as 
    \begin{align}
        {\rm Pr}(Y \neq \widehat{Y}) &= \frac{1}{4 \pi} \int_{0}^{\infty} \int_{-\infty}^{\infty}e^{-\frac{(t - \alpha \tanh{(\beta x)})^2 + (x + \beta )^2 }{2}}  dxdt + \frac{1}{4 \pi} \int_{-\infty}^{0} \int_{-\infty}^{\infty}e^{-\frac{(t - \alpha \tanh{(\beta x)})^2 + (x - \beta )^2 }{2}}  dxdt\nonumber \\
        &=\frac{1}{4 \pi} \int_{0}^{\infty} \int_{-\infty}^{\infty}e^{-\frac{(t - \alpha \tanh{(\beta x)})^2 + (x + \beta )^2 }{2}}  dxdt + \frac{1}{4 \pi} \int_{0}^{\infty} \int_{-\infty}^{\infty}e^{-\frac{(t + \alpha \tanh{(\beta x)})^2 + (x - \beta )^2 }{2}}  dxdt \nonumber \\
        &= \frac{1}{2 \pi} \int_{0}^{\infty} \int_{-\infty}^{\infty} e^{-\frac{t^2 + \alpha^2 \tanh^2{(\beta x)} + x^2 +\beta^2}{2}} \cosh(t \alpha \tanh{(\beta x)} - \beta x)  dxdt \label{eq: error_c} 
    \end{align}

\section{Further Discussions on Limiting Cases in Remark \ref{rem:limiting_cases_soft}}
\label{ch:extreme_points}
In Figure~\ref{fig:beta_one_dim_22}, we present, following the discussions in Remark~\ref{rem:limiting_cases_soft}, numerical behaviors of the two proposed lower bounds at the extreme points where $R$ is rather large and $R=0$, for both $\beta=1$ and $\beta=\sqrt{2}$. 
We observe that:
\begin{itemize}
    \item[(i)] for $R = 0$ nats, we have $\alpha_{\text{lb}_1}= 0$ (per its definition in \eqref{eq:LB_1} as already discussed in Remark~\ref{rem:limiting_cases_soft}, so that $I_3(\alpha_{\text{lb}_1}) = 0$; and
    \item[(ii)] as $R \rightarrow \infty$ nats, we have that both $\alpha_{\text{lb}_1}$ and $\alpha_{\text{lb}_2}$ reach infinity, so that both lower bounds $I_3(\alpha_{\text{lb}_1})$ and $I_4(\alpha_{\text{lb}_2}) $ converge to the optimal point of $I(X; Y)$.
\end{itemize}
This thus provides numerical evidence for the statement made in Remark~\ref{rem:limiting_cases_soft}.

\begin{figure}[htb]
\centering
        \begin{subfigure}[b]{0.4\textwidth}
        \begin{tikzpicture}[scale=0.6]
\pgfplotsset{every major grid/.append style={densely dashed}}
\begin{axis}[%
width=3.528in,
height=3.029in,
at={(1.011in,0.767in)},
scale only axis,
xmin=0,
xmax=40,
xlabel style={font=\color{white!30!black}},
xlabel={$I(X; T)$ (nats)},
ymin=0,
ymax=0.35,
ylabel style={font=\color{white!15!black}},
ylabel={$I(Y; T)$ (nats)},
axis background/.style={fill=white},
xmajorgrids,
ymajorgrids,
legend style={at={(0.725,0. 0)}, anchor=south west, legend cell align=left, align=left, draw=white!15!black}
]
\addplot [color=mycolor1, line width=2.5pt, mark size=3.0pt, mark=+, mark options={solid, mycolor1}]
  table[row sep=crcr]{%
0	0.3368\\
2	0.3368\\
5	0.3368\\
10	0.3368\\
15	0.3368\\
20	0.3368\\
40	0.3368\\
};
\addlegendentry{$I(X; Y)$}

\addplot [color=mycolor2, line width=2.5pt, mark size=3.0pt, mark=o, mark options={solid, mycolor2}]
  table[row sep=crcr]{%
0	8.88178419700125e-16\\
2	0.286880886827405\\
5	0.317537820181277\\
10	0.326024326222184\\
15	0.328984493223333\\
20	0.330552474900515\\
40	0.332837877073753\\
};
\addlegendentry{$I_3(\alpha_{\text{lb}_1})$}

\addplot [color=mycolor3, line width=2.5pt, mark size=3.0pt, mark=triangle, mark options={solid,  rotate =180, mycolor3}]
  table[row sep=crcr]{%
2	0.293794664346162\\
5	0.318475062770588\\
10	0.326276309470357\\
15	0.329106537870612\\
20	0.330625003451726\\
40	0.333054444880549\\
};
\addlegendentry{$I_4(\alpha_{\text{lb}_2})$}
\end{axis}
\end{tikzpicture}%
        \caption{$\beta=1$.}
        \label{fig:beta_1_22}
        \end{subfigure}
        ~
        \begin{subfigure}[b]{0.4\textwidth}
        \begin{tikzpicture}[scale=0.6]
        \pgfplotsset{every major grid/.append style={densely dashed}}
        \begin{axis}[%
        width=3.528in,
        height=3.029in,
        at={(1.011in,0.767in)},
        scale only axis,
        xmin=0,
        xmax=20,
        xlabel style={font=\color{white!30!black}},
        xlabel={$I(X; T)$ (nats)},
        ymin=0,
        ymax=0.51,
        ylabel style={font=\color{white!15!black}},
        ylabel={$I(Y; T)$ (nats)},
        axis background/.style={fill=white},
        xmajorgrids,
        ymajorgrids,
        legend style={at={(0.725,0.0)}, anchor=south west, legend cell align=left, align=left, draw=white!15!black}
        ]
        
        \addplot [color=mycolor1, line width=2.5pt, mark size=3.0pt, mark=+, mark options={solid, mycolor1}]
          table[row sep=crcr]{%
        0	0.5\\
        2	0.5\\
        5	0.5\\
        10	0.5\\
        15	0.5\\
        20	0.5\\
        };
        \addlegendentry{$I(X; Y)$}
       
        \addplot [color=mycolor2, line width=2.5pt, mark size=3.0pt, mark=o, mark options={solid, mycolor2}]
          table[row sep=crcr]{%
        0	0\\
        2	0.4700147298398\\
        5	0.486803275464515\\
        10	0.491878274709773\\
        15	0.49451359957651\\
        20	0.495620362972182\\
        };
        \addlegendentry{$I_3(\alpha_{\text{lb}_1})$}
        
        \addplot [color=mycolor3, line width=2.5pt, mark size=3.0pt, mark=triangle, mark options={solid, rotate=180, mycolor3}]
          table[row sep=crcr]{%
        2	0.474172671520654\\
        5	0.487358859359633\\
        10	0.492065500834975\\
        15	0.494666149778046\\
        20	0.49495284187371\\
        };
        \addlegendentry{$I_4(\alpha_{\text{lb}_2})$}    
        \end{axis}
        \end{tikzpicture}%
        \caption{$\beta=\sqrt{2}$.}
        \label{fig:beta_sqrt2_22}
        \end{subfigure}
        
 \caption{The objective mutual information $I(Y; T)$ versus the constraint $I(X; T)$ for two proposed lower bounds $I_3(\alpha_{\text{lb}_1})$ and $I_4(\alpha_{\text{lb}_2})$ when $\beta \in \{1, \sqrt{2}\}$.}
     \label{fig:beta_one_dim_22}
\end{figure}
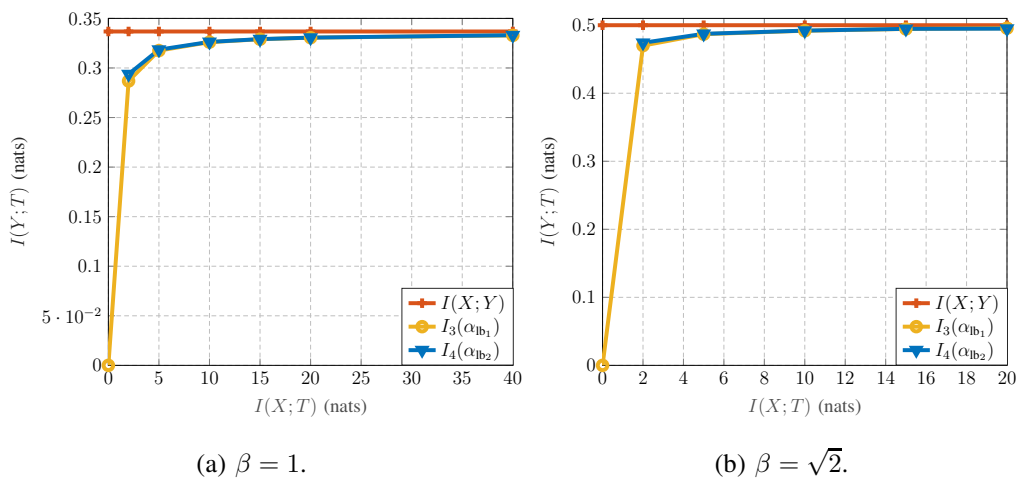

\section{MNIST Data Pre-processing}
\label{sec:MNIST_processing}
Recall that our theoretical results assume that the input data $\x \in \RR^{d_0}$ are drawn from the following symmetric binary Gaussian mixture model
\begin{equation}\label{eq:GMM}
  \mathcal C_1: \x \sim \mathcal N(-\bbeta, \I_{d_0}), \quad \mathcal C_2: \x \sim \mathcal N(+\bbeta, \I_{d_0}).
\end{equation}
For vectorized MNIST images of dimension $p = 784$ composed of ten classes (number $0$ to $9$), here we choose the images of number $7$ versus $9$ to perform binary classification. 
For the sake of computational complexity, we apply a random projection that reduce the $784$-dimensional raw data vector $\tilde \x$ to obtain a three-dimensional feature $\x$, i.e., $\x = \mathbf{W} \tilde \x \in \RR^3$, with the i.i.d.\@ entries of $\mathbf{W} \in \RR^{3 \times 783}$ following a standard Gaussian distribution. 
Then, we collect three-dimensional feature matrices $\X_1 \in \RR^{3 \times n_1}$ and $\X_2 \in \RR^{3 \times n_2}$ of class $\mathcal C_1$ and $\mathcal C_2$, and we perform further pre-processing to make them closer to \eqref{eq:GMM}. 
First, the empirical means of each class are computed as $\hat \bmu_1 = \frac{1}{n_1} \X_1 \one_{n_1}$ and $\hat \bmu_2 =\frac{1}{n_2} \X_2 \one_{n_2}$. 
We then compute the empirical covariances as $\hat \C_1 = \frac1{n_1} (\X_1 - \hat \mu_1 \one_{n_1}^\T) (\X_1 - \hat \mu_1 \one_{n_1}^\T)^\T$ and similarly for $\X_2$. 
Finally, whitened features matrices are obtained via
\begin{equation}\label{eq:white}
    \tilde \X_1 = \frac12 (\hat \bmu_1 - \hat \bmu_2) + \hat \C_1^{-\frac12} (\X_1 - \hat \bmu_1 \one_{n_1}^\T), 
\end{equation}
for class $\mathcal C_1$ and similarly $\tilde \X_2$ for class $\mathcal C_1$. {In the simulation, we choose $2000$ samples of each class to estimate mutual information using Jackknife approach.}

\bibliographystyle{IEEEtran}
\bibliography{IEEEabrv,liao}

\begin{thebibliography}{10}
\providecommand{\url}[1]{#1}
\csname url@samestyle\endcsname
\providecommand{\newblock}{\relax}
\providecommand{\bibinfo}[2]{#2}
\providecommand{\BIBentrySTDinterwordspacing}{\spaceskip=0pt\relax}
\providecommand{\BIBentryALTinterwordstretchfactor}{4}
\providecommand{\BIBentryALTinterwordspacing}{\spaceskip=\fontdimen2\font plus
\BIBentryALTinterwordstretchfactor\fontdimen3\font minus
  \fontdimen4\font\relax}
\providecommand{\BIBforeignlanguage}[2]{{%
\expandafter\ifx\csname l@#1\endcsname\relax
\typeout{** WARNING: IEEEtran.bst: No hyphenation pattern has been}%
\typeout{** loaded for the language `#1'. Using the pattern for}%
\typeout{** the default language instead.}%
\else
\language=\csname l@#1\endcsname
\fi
#2}}
\providecommand{\BIBdecl}{\relax}
\BIBdecl

\bibitem{tishby2000information}
N.~Tishby, F.~C. Pereira, and W.~Bialek, ``{The information bottleneck
  method},'' \emph{arXiv}, 2000.

\bibitem{zaidientropy2020}
A.~Zaidi, I.~Estella-Aguerri, and S.~Shamai, ``On the information bottleneck
  problems: Models, connections, applications and information theoretic
  views,'' \emph{Entropy}, vol.~22, no.~2, p. 151, Feb. 2020.

\bibitem{Hassanpour2017overview}
S.~Hassanpour, D.~Wuebben, and A.~Dekorsy, ``Overview and investigation of
  algorithms for the information bottleneck method,'' in \emph{Proc. 11th Int.
  ITG Conf. Syst., Commun. Coding (SCC)}, Hamburg, Germany, Feb. 2017, pp.
  1--6.

\bibitem{alemi2017deep}
A.~A. Alemi, I.~Fischer, J.~V. Dillon, and K.~Murphy, ``Deep variational
  information bottleneck,'' in \emph{ICLR}, 2017.

\bibitem{goldfeld2020information}
Z.~Goldfeld and Y.~Polyanskiy, ``The information bottleneck problem and its
  applications in machine learning,'' \emph{IEEE J. Sel. Areas Inf. Theory},
  vol.~1, no.~1, pp. 19--38, May 2020.

\bibitem{dobrushin1962information}
R.~Dobrushin and B.~Tsybakov, ``Information transmission with additional
  noise,'' \emph{IRE Trans. Theory}, vol.~8, no.~5, pp. 293--304, Sep. 1962.

\bibitem{witsenhausen1975conditional}
H.~Witsenhausen and A.~Wyner, ``A conditional entropy bound for a pair of
  discrete random variables,'' \emph{IEEE Trans. Inf. Theory}, vol.~21, no.~5,
  pp. 493--501, Sep. 1975.

\bibitem{courtade2013multiterminal}
T.~A. Courtade and T.~Weissman, ``Multiterminal source coding under logarithmic
  loss,'' \emph{IEEE Trans. Inf. Theory}, vol.~60, no.~1, pp. 740--761, Jan.
  2014.

\bibitem{sanderovich2008communication}
A.~Sanderovich, S.~Shamai, Y.~Steinberg, and G.~Kramer, ``Communication via
  decentralized processing,'' \emph{IEEE Trans. Inf. Theory}, vol.~54, no.~7,
  pp. 3008--3023, July 2008.

\bibitem{aguerri2019capacity}
I.~E. Aguerri, A.~Zaidi, G.~Caire, and S.~S. Shitz, ``On the capacity of cloud
  radio access networks with oblivious relaying,'' \emph{IEEE Trans. Inf.
  Theory}, vol.~65, no.~7, pp. 4575--4596, July 2019.

\bibitem{xu2021information}
H.~Xu, T.~Yang, G.~Caire, and S.~Shamai, ``Information bottleneck for a
  {Rayleig} fading {MIMO} channel with an oblivious relay,''
  \emph{Information}, vol.~12, no.~4, p. 155, April 2021.

\bibitem{xuhao2021information}
H.~Xu, T.~Yang, G.~Caire, and S.~S. Shitz, ``Information bottleneck for an
  oblivious relay with channel state information: the vector case,'' pp.
  2483--2488, Melbourne, Australia, Jul. 2021.

\bibitem{xu2022distributed}
H.~Xu, K.-K. Wong, G.~Caire, and S.~S. Shitz, ``Distributed information
  bottleneck for a primitive {Gaussian} diamond channel with {Rayleigh}
  fading,'' in \emph{Proc. IEEE Int. Symp. Inf. Theory (ISIT)}, Espoo, Finland,
  Jun. 2022, pp. 2845--2850.

\bibitem{song2023distributed}
Y.~Song, H.~Xu, K.-K. Wong, G.~Caire, and S.~S. Shitz, ``Distributed
  information bottleneck for a primitive gaussian diamond {MIMO} channel,'' in
  \emph{Proc. IEEE Int. Symp. Inf. Theory (ISIT)}, Taipei, Taiwan, June 2023,
  pp. 1484--1489.

\bibitem{bengio2013representation}
Y.~Bengio, A.~Courville, and P.~Vincent, ``Representation learning: A review
  and new perspectives,'' \emph{IEEE Trans. Pattern Anal. Mach. Intell.},
  vol.~35, no.~8, p. 1798–1828, aug 2013.

\bibitem{achille2018information}
A.~Achille and S.~Soatto, ``Information dropout: Learning optimal
  representations through noisy computation,'' \emph{IEEE Trans. Pattern Anal.
  Mach. Intell.}, vol.~40, no.~12, pp. 2897--2905, Dec. 2018.

\bibitem{bishop2006pattern}
C.~M. Bishop, \emph{{Pattern Recognition and Machine Learning}}, 1st~ed.\hskip
  1em plus 0.5em minus 0.4em\relax Springer-Verlag New York, 2006.

\bibitem{shwartz2017opening}
R.~Shwartz-Ziv and N.~Tishby, ``Opening the black box of deep neural networks
  via information,'' \emph{arXiv preprint arXiv:1703.00810}, 2017.

\bibitem{kim2021distilling}
J.~Kim, B.-K. Lee, and Y.~M. Ro, ``Distilling robust and non-robust features in
  adversarial examples by information bottleneck,'' \emph{Advances in Neural
  Information Processing Systems}, vol.~34, pp. 17\,148--17\,159, 2021.

\bibitem{rishby2015deep}
N.~Tishby and N.~Zaslavsky, ``Deep learning and the information bottleneck
  principle,'' in \emph{Proc. IEEE Inf. Theory Workshop (ITW)}, April 2015, pp.
  1--5.

\bibitem{saxe2019information}
A.~M. Saxe, Y.~Bansal, J.~Dapello, M.~Advani, A.~Kolchinsky, B.~D. Tracey, and
  D.~D. Cox, ``On the information bottleneck theory of deep learning,''
  \emph{ICLR}, no.~12, pp. 368--377, April 2017.

\bibitem{dai2018compressing}
B.~Dai, C.~Zhu, B.~Guo, and D.~Wipf, ``Compressing neural networks using the
  variational information bottleneck,'' in \emph{ICML}, 2018, pp. 1135--1144.

\bibitem{hanineq}
T.~M. Cover and J.~A. Thomas, \emph{Elements of Information Theory (2nd
  edition)}.\hskip 1em plus 0.5em minus 0.4em\relax Hoboken, NJ: Wiley, 2006.

\bibitem{blahut1972computation}
R.~Blahut, ``Computation of channel capacity and rate-distortion functions,''
  \emph{IEEE Trans. Inf. Theory}, vol.~18, no.~4, pp. 460--473, 1972.

\bibitem{guo2005mutual}
D.~Guo, S.~Shamai, and S.~Verd{\'u}, ``Mutual information and minimum
  mean-square error in gaussian channels,'' \emph{IEEE transactions on
  information theory}, vol.~51, no.~4, pp. 1261--1282, 2005.

\bibitem{slonim1999Agglomerative}
N.~Slonim and N.~Tishby, ``Agglomerative {Information Bottleneck},'' in
  \emph{Adv. Neural Inf. Process.}, vol.~12, 1999.

\bibitem{slonim2002unsupervised}
N.~Slonim, N.~Friedman, and N.~Tishby, ``Unsupervised document classification
  using sequential information maximization,'' in \emph{Proc. 25th Ann. Int’l
  ACM SIGIR Conf. Research and Development in Information Retrieval (SIGIR)},
  2002, pp. 129--136.

\bibitem{strouse2017deterministic}
D.~Strouse and D.~J. Schwab, ``The deterministic information bottleneck,''
  \emph{Neural computation}, vol.~29, no.~6, pp. 1611--1630, 2017.

\bibitem{hornik1991approximation}
K.~Hornik, ``Approximation capabilities of multilayer feedforward networks,''
  \emph{Neural Networks}, vol.~4, no.~2, pp. 251--257, 1991.

\bibitem{lecun1998gradient}
Y.~LeCun, L.~Bottou, Y.~Bengio, and P.~Haffner, ``Gradient-based learning
  applied to document recognition,'' \emph{Proc. IEEE}, vol.~86, no.~11, pp.
  2278--2324, 1998.

\bibitem{johnson1984extensions}
W.~B. Johnson and J.~Lindenstrauss, ``{Extensions of Lipschitz mappings into a
  Hilbert space},'' \emph{Contemporary Mathematics}, pp. 189--206, 1984.

\bibitem{drineas2016randnla}
P.~Drineas and M.~W. Mahoney, ``{RandNLA: randomized numerical linear
  algebra},'' \emph{Communications of the ACM}, vol.~59, no.~6, pp. 80--90,
  2016.

\bibitem{Jackknife2011zeng}
X.~Zeng, Y.~Xia, and H.~Tong, ``Jackknife approach to the estimation of mutual
  information,'' \emph{Proceedings of the National Academy of Sciences}, vol.
  115, no.~40, pp. 9956--9961, 2018.

\end{thebibliography}

\end{document}